\documentclass[english]{article}
\usepackage[T1]{fontenc}
\usepackage[latin9]{inputenc}
\usepackage{geometry}
\usepackage{array}
\usepackage{float}
\usepackage{mathtools}
\usepackage{multirow}
\usepackage{amsmath}
\usepackage{amsthm}
\usepackage{amssymb}
\usepackage{graphicx}
\usepackage[authoryear,round]{natbib}

\makeatletter

\providecommand{\tabularnewline}{\\}
\floatstyle{ruled}
\newfloat{algorithm}{tbp}{loa}
\providecommand{\algorithmname}{Algorithm}
\floatname{algorithm}{\protect\algorithmname}

\theoremstyle{plain}
\newtheorem{thm}{\protect\theoremname}[section]
\theoremstyle{plain}
\newtheorem{lem}[thm]{\protect\lemmaname}
\theoremstyle{plain}
\newtheorem{cor}[thm]{\protect\corollaryname}
\theoremstyle{remark}
\newtheorem{rem}[thm]{\protect\remarkname}
\theoremstyle{definition}
\newtheorem{defn}[thm]{\protect\definitionname}
\theoremstyle{remark}
\newtheorem{claim}[thm]{\protect\claimname}

\usepackage{hyperref}

\@ifundefined{showcaptionsetup}{}{%
 \PassOptionsToPackage{caption=false}{subfig}}
\usepackage{subfig}
\makeatother

\usepackage{babel}
\providecommand{\claimname}{Claim}
\providecommand{\corollaryname}{Corollary}
\providecommand{\definitionname}{Definition}
\providecommand{\lemmaname}{Lemma}
\providecommand{\remarkname}{Remark}
\providecommand{\theoremname}{Theorem}

\begin{document}
\global\long\def\smax{\mathrm{smax}}%
\global\long\def\smin{\mathrm{smin}}%
\global\long\def\R{\mathbb{R}}%
\global\long\def\D{\mathrm{diag}}%
\global\long\def\opt{\mathrm{OPT}}%
\global\long\def\L{\mathcal{L}}%
\global\long\def\E{\mathbb{E}}%
\global\long\def\proj{\mathrm{proj}}%

\title{Multiplicative Weights Update, Area Convexity and Random Coordinate
Descent for Densest Subgraph Problems}
\author{Ta Duy Nguyen\thanks{Department of Computer Science, Boston University, \texttt{taduy@bu.edu}.} \and
Alina Ene\thanks{Department of Computer Science, Boston University, \texttt{aene@bu.edu}.}}
\date{}
\maketitle
\begin{abstract}
We study the densest subgraph problem and give algorithms via multiplicative
weights update and area convexity that converge in $O\left(\frac{\log m}{\epsilon^{2}}\right)$
and $O\left(\frac{\log m}{\epsilon}\right)$ iterations, respectively,
both with nearly-linear time per iteration. Compared with the work
by \citet{bahmani2014efficient}, our MWU algorithm uses a very different
and much simpler procedure for recovering the dense subgraph from
the fractional solution and does not employ a binary search. Compared
with the work by \citet{boob2019faster}, our algorithm via area convexity
improves the iteration complexity by a factor $\Delta$---the maximum
degree in the graph, and matches the fastest theoretical runtime currently
known via flows \citep{chekuri2022densest} in total time. Next, we
study the dense subgraph decomposition problem and give the first
practical iterative algorithm with linear convergence rate $O\left(mn\log\frac{1}{\epsilon}\right)$
via accelerated random coordinate descent. This significantly improves
over $O\left(\frac{m\sqrt{mn\Delta}}{\epsilon}\right)$ time of the
FISTA-based algorithm by \citet{harb2022faster}. In the high precision
regime $\epsilon\ll\frac{1}{n}$ where we can even recover the exact
solution, our algorithm has a total runtime of $O\left(mn\log n\right)$,
matching the  exact algorithm via parametric flows \citep{gallo1989fast}.
Empirically, we show that this algorithm is very practical and scales
to very large graphs, and its performance is competitive with widely
used methods that have significantly weaker theoretical guarantees.

\end{abstract}

\section{Introduction}

In this work, we study the densest subgraph problem (DSG) and its
generalization to finding dense subgraph decompositions of graphs.
In the densest subgraph problem, we are given a graph $G=(V,E)$ and
the goal is to find a subgraph of maximum density $\left|E(S)\right|/\left|S\right|$,
where $\left|E(S)\right|$ is the number of edges in the graph induced
by $S$ (for weighted graphs, we consider the total weight of the
edges). Densest subgraphs and related dense subgraph discovery problems
have seen numerous applications in machine learning and data mining,
including DNA motif detection, fraud detection, and distance query
computation (see \citep{lee2010survey,gionis2015dense,farago2019search,tsourakakisdense,lanciano2023survey}
for more comprehensive surveys).

The densest subgraph problem and its generalizations are fundamental
graph optimization problems with a long history in algorithm design.
A classical result due to \citet{goldberg1984finding} showed that
DSG can be solved in polynomial time via a reduction to maximum flow.
Specifically, \citet{goldberg1984finding} showed that, given a guess
$D$ for the maximum density, one can either find a subgraph with
density at least $D$ or certify that no such subgraph exists by computing
a maximum $s$-$t$ flow in a suitably defined network. This approach
together with binary search allows us to compute an optimal solution
using a logarithmic number of maximum flow computations. \citet{gallo1989fast}
designed a more efficient algorithm for DSG via a reduction to parametric
maximum flows that solve a sequence of related maximum flow instances
more efficiently than the binary search approach. This led to an algorithm
for DSG with running time $O\left(nm\log\left(n^{2}/m\right)\right)$,
where $n$ and $m$ are the number of nodes and edges in the input
graph, respectively. Based on the near-linear time algorithm for computing
minimum-cost flows by \citet{chen2022maximum}, \citet{harb2022faster}
gave an algorithm that computes an optimal dense decomposition in
$O\left(m^{1+o(1)}\right)$ time. The recent work of \citet{chekuri2022densest}
designed a maximum flow based algorithm that computes an $\left(1-\epsilon\right)$-approximate
solution in time $O\left(\frac{m\log^{2}\left(m\right)}{\epsilon}\right)$.
To the best of our knowledge, these are the fastest running times
for exact and approximate algorithms, respectively.

The maximum/minimum-cost flow based approaches provide a rich theoretical
framework for developing algorithms with provable guarantees for DSG
and related problems. Although these algorithms have strong theoretical
guarantees, their practical performance and scalability is more limited
and they do not scale to very large graphs \citep{boob2020flowless}.
Moreover, these algorithms are inherently sequential and they cannot
leverage parallel and distributed computation \citep{bahmani2012densest,bahmani2014efficient}. 

The aforementioned limitations of the flow-based algorithms have motivated
the development of iterative algorithms based on linear and convex
programming formulations. This line of work has led to the development
of iterative algorithms based on continuous optimization frameworks
such as multiplicative weights update \citep{bahmani2014efficient},
Frank-Wolfe \citep{danisch2017large,harb2023convergence}, and accelerated
gradient descent \citep{harb2022faster}. These frameworks have also
inspired Greedy algorithms that are combinatorial and very efficient
in practice: the Greedy peeling algorithm \citep{charikar2000greedy}
that makes a single pass over the graph but it achieves only a $1/2$
approximation, and a variant of it called Greedy++ \citep{boob2020flowless}
that makes multiple passes but it achieves a $1-\epsilon$ approximation
for any target approximation error $\epsilon$. Subsequent work established
theoretical convergence guarantees for Greedy++ \citep{chekuri2022densest}
and showed that it is equivalent to a Frank-Wolfe algorithm \citep{harb2023convergence}.

Despite the wide range of algorithms designed to solve DSG and its
generalizations, there still remain important directions for improvement
in theory and in practice. The Greedy peeling and Greedy++ algorithms
are very efficient in practice, but they are also inherently sequential
and their theoretical guarantees are weaker than the iterative algorithms
based on continuous optimization. On the other hand, the practical
applicability of the latter algorithms is significantly more limited.
The algorithm of \citet{bahmani2014efficient} uses more complex subroutines,
including a binary search over the optimal solution value and an involved
procedure for constructing the primal solution (the dense subgraph).
The algorithm of \citet{boob2019faster} only provides (an approximation
to) the solution value, and not the solution itself. Moreover, the
number of iterations of these algorithms also depends polynomially
in the maximum degree of the graph and/or the number of edges, which
can be prohibitive for large graphs with nodes of very high degree.
An important direction is to obtain algorithms with stronger theoretical
convergence guarantees that enjoy fast convergence with simple iterations
that are easily parallelizable and very efficient in theory and in
practice.

Another limitation is that the iterative algorithms only construct
approximate solutions and they require $\mathrm{poly(1/\epsilon)}$
iterations to achieve a $1-\epsilon$ approximation. As a result,
the running time can be prohibitively large for obtaining very good
approximation guarantees. An important direction is to obtain scalable
and practical iterative algorithms with a much more beneficial $\log(1/\epsilon)$
dependence on the approximation error. Such algorithms would allow
for obtaining exact solutions (by setting $\epsilon$ polynomially
small in the size of the graph, and thus incurring only a logarithmic
factor in the running time). Currently, the only exact algorithms
known are based on maximum/minimum-cost flow and they are prohibitive
in practice as discussed above.

The aforementioned directions are the main motivation behind this
work, and we make several contributions towards resolving them as
we outline below.

\subsection{Contributions}

Building on the algorithm of \citet{bahmani2014efficient}, we give
an iterative algorithm based on the multiplicative weights update
framework (MWU, \citet{arora2012multiplicative}) that converges in
$O\big(\frac{\log m}{\epsilon^{2}}\big)$ iterations. Each iteration
of our algorithm can be implemented very efficiently in nearly-linear
time, and it can be easily parallelized by processing each vertex
and its incident edges in parallel on separate machines. Through a
combination of the techniques in the work of \citet{bahmani2014efficient}
as well as novel insights and techniques we introduce, we are able
to preserve all of the strengths of the result of \citet{bahmani2014efficient}:
compared to other approaches, the number of iterations is independent
of the maximum degree $\Delta$ of the graph (in contrast, all other
approaches incur an extra $\Delta$ factor); the algorithm can be
applied to many different settings, including to streaming, parallel,
and distributed computation \citep{bahmani2012densest,bahmani2014efficient,su2019distributed},
as well as differentially private algorithms \citep{dhulipala2022differential}.
Simultaneously, we significantly strengthen and simplify the prior
approach, and remove its main limitations: we design a very different
algorithm for constructing a primal solution (the dense subgraph)
from the fractional solution to a modified dual problem that the MWU
algorithm constructs; we give an end-to-end algorithm for implementing
each iteration that does not employ a binary search. Due to the wide
range of applications of this framework mentioned earlier, we expect
that our improved approach will lead to improvements in all of these
diverse settings and for other related problems and beyond.

Our second contribution builds on our MWU algorithm and the area convexity
technique introduced by \citet{sherman2017area} for flow problems
and extensions and further utilized by \citet{boob2019faster} for
solving packing and covering LPs and the densest subgraph problem.
By replacing the entropy regularizer with an area convex regularizer,
we design an iterative algorithm with an improved iteration complexity
of $O\big(\frac{\log m}{\epsilon}\big)$ and a nearly-linear time
per iteration. Our algorithm improves upon the result of \citet{boob2019faster}
by a factor $\Delta$ (the maximum degree in the graph), both in the
number of iterations and overall running time. Furthermore, we show
how to construct the primal solution (the dense subgraph), whereas
\citet{boob2019faster} can only output the value of the solution.
Similarly to prior work based on the area convexity technique \citep{sherman2017area,boob2019faster},
each iteration of is more complex and less practical than our MWU-based
approach. However, the result is theoretically interesting for at
least two reasons: it shows that the $\frac{1}{\epsilon^{2}}$ barrier
for entropy-based MWU algorithms can be overcome without incurring
a polynomial factor loss in the iteration complexity; and the overall
running time matches that of the flow-based algorithm of \citet{chekuri2022densest},
which achieves the fastest theoretical running time currently known
for densest subgraph but is inherently sequential. In contrast to
flow-based algorithms, area convexity is closely related to practical
(extra)gradient methods \citep{jambulapati2023revisiting} and has
found successful applications beyond DSG, including solving structured
LPs \citep{boob2019faster}, optimal transport \citep{jambulapati2019direct}
and matching \citep{assadi2022semi}. Improvements in DSG could potentially
be used as an example to derive new iterative frameworks for other
continuous and combinatorial problems.

Finally, by adapting the approach of\textbf{ }\citet{ene2015random,ene2017decomposable}\textbf{
}for minimizing submodular functions with a decomposable structure,
we obtain the first practical iterative algorithms for DSG and generalizations
with a $\log\frac{1}{\epsilon}$ dependency on the approximation error
$\epsilon$. Similarly to \citet{harb2022faster}, our algorithms
solve a convex programming formulation that captures DSG and its generalization
to finding a dense subgraph decomposition (we defer the definitions
to Section \ref{sec:Preliminaries}). The objective function of the
convex program is smooth with smoothness parameter proportional to
the maximum degree $\Delta$, but importantly it is not strongly convex.
\citet{harb2022faster} used the accelerated FISTA algorithm to solve
the convex program, and obtained a running time of $O\left(\frac{m\sqrt{mn\Delta}}{\epsilon}\right)$.
In contrast, we adapt the coordinate descent algorithm and its accelerated
version developed by \citet{ene2015random} for submodular minimization.
Our accelerated algorithm achieves a running time of $O\left(mn\log\frac{1}{\epsilon}\right)$
in expectation. Crucially, we achieve an exponentially improved dependence
on $1/\epsilon$ (i.e., a linear convergence rate) despite the lack
of strong convexity in the objective, by leveraging the combinatorial
structure as in\textbf{ }\citet{ene2015random,ene2017decomposable}.
Additionally, the objective has constant smoothness in each coordinate
(in contrast to the $\Delta$ global smoothness), leading to further
improvements in the running time. In the high precision regime $\epsilon\ll\frac{1}{n}$
where we can even recover the exact solution, our accelerated algorithm
has a total runtime of $O\left(mn\log n\right)$, matching the algorithm
via parametric flows by \citet{gallo1989fast}. Although this does
not match the state of the art algorithm via minimum-cost flows by
\citet{harb2022faster}, in contrast to these flow-based algorithms,
our algorithms are very simple and easy to implement. Our experimental
evaluation shows that our algorithms are very practical and scalable
to very large graphs, and are competitive with the highly practical
Greedy++ algorithm while enjoying significantly stronger theoretical
guarantees.

We show comparisons of runtime between existing methods and our algorithms
in Table \ref{table:comparison}. 

\begin{table*}
\centering{}{\small{}\caption{{\small{}Comparison between existing algorithms for computing approximate
densest subgraphs/densest subgraph decomposition. $m,n,\Delta$ are
the number of edges and vertices and the maximum degree in the graph.
$\protect\opt$ is the maximum density of a subgraph.}}
\label{table:comparison}}%
\begin{tabular}{>{\raggedright}m{0.34\textwidth}>{\centering}m{0.1\textwidth}>{\centering}m{0.13\textwidth}>{\raggedright}m{0.31\textwidth}}
\hline 
\noalign{\vskip\doublerulesep}
{\small{}Algorithm} & {\small{}No. of Iter.} & {\small{}Per iter.} & {\small{}Note}\tabularnewline[\doublerulesep]
\hline 
\noalign{\vskip\doublerulesep}
\hline 
\noalign{\vskip\doublerulesep}
{\small{}Greedy++ \citep{boob2020flowless}} & {\small{}$O\big(\frac{\Delta\log m}{\opt\epsilon^{2}}\big)$} & {\small{}$O(m\log n)$} & \tabularnewline[\doublerulesep]
\hline 
\noalign{\vskip\doublerulesep}
\noalign{\vskip\doublerulesep}
{\small{}Frank-Wolfe \citep{danisch2017large}} & {\small{}$O\big(\frac{m\Delta}{\epsilon^{2}}\big)$} & {\small{}$O(m)$} & \tabularnewline[\doublerulesep]
\hline 
\noalign{\vskip\doublerulesep}
\noalign{\vskip\doublerulesep}
{\small{}\citet{chekuri2022densest} (flow-based)} & {\small{}$O\big(\frac{\log m}{\epsilon}\big)$} & {\small{}$O(m\log m)$} & {\small{}Based on blocking flows}\tabularnewline[\doublerulesep]
\hline 
\noalign{\vskip\doublerulesep}
\noalign{\vskip\doublerulesep}
{\small{}\citet{gallo1989fast} (flow-based, exact)} & \multicolumn{2}{c}{{\small{}$O\left(nm\log\left(n^{2}/m\right)\right)$ (total)}} & {\small{}Based on push-relabel}\tabularnewline[\doublerulesep]
\hline 
\noalign{\vskip\doublerulesep}
\noalign{\vskip\doublerulesep}
{\small{}\citet{harb2022faster} (flow-based, exact)} & \multicolumn{2}{c}{{\small{}$O(m^{1+o(1)})$ (total)}} & {\small{}Based on min-cost flow algorithm by }\citet{chen2022maximum}\tabularnewline[\doublerulesep]
\hline 
\noalign{\vskip\doublerulesep}
\noalign{\vskip\doublerulesep}
{\small{}\citet{bahmani2014efficient}} & {\small{}$O\big(\frac{\log m}{\epsilon^{2}}\big)$} & {\small{}$O(m)$} & \multirow{2}{0.31\textwidth}{{\small{}New and simpler construction of the solution; Remove binary
search}}\tabularnewline[\doublerulesep]
\cline{1-3} \cline{2-3} \cline{3-3} 
\noalign{\vskip\doublerulesep}
\noalign{\vskip\doublerulesep}
\textbf{\small{}Algorithm \ref{alg:mwu} (ours)} & {\small{}$O\big(\frac{\log m}{\epsilon^{2}}\big)$} & {\small{}$O(m\log\Delta)$} & \tabularnewline[\doublerulesep]
\hline 
\noalign{\vskip\doublerulesep}
\noalign{\vskip\doublerulesep}
{\small{}\citet{boob2019faster} (solution value only)} & {\small{}$O\big(\frac{\Delta\log m}{\epsilon}\big)$} & {\small{}$O(m\log\frac{1}{\epsilon})$} & \multirow{2}{0.31\textwidth}{{\small{}Improve a factor $\Delta$ and construct solution}}\tabularnewline[\doublerulesep]
\cline{1-3} \cline{2-3} \cline{3-3} 
\noalign{\vskip\doublerulesep}
\noalign{\vskip\doublerulesep}
\textbf{\small{}Algorithm \ref{alg:area-convexity} (ours)} & {\small{}$O\big(\frac{\log m}{\epsilon}\big)$} & {\small{}$O\big(m\log\Delta\log\frac{1}{\epsilon}\big)$} & \tabularnewline[\doublerulesep]
\hline 
\noalign{\vskip\doublerulesep}
\noalign{\vskip\doublerulesep}
{\small{}\citet{harb2022faster} (additive error)} & {\small{}$O\big(\frac{\sqrt{mn\Delta}}{\epsilon}\big)$} & {\small{}$O(m)$} & \multirow{2}{0.31\textwidth}{{\small{}Improve total time by a factor at least $\frac{\sqrt{\Delta}}{\epsilon\log\frac{n}{\epsilon}}$}}\tabularnewline[\doublerulesep]
\cline{1-3} \cline{2-3} \cline{3-3} 
\noalign{\vskip\doublerulesep}
\noalign{\vskip\doublerulesep}
\textbf{\small{}Algorithm \ref{alg:acdm} (ours) }{\small{}(in expectation)} & {\small{}$O\big(mn\log\frac{n}{\epsilon}\big)$} & \textbf{\small{}$O(1)$} & \tabularnewline[\doublerulesep]
\hline 
\noalign{\vskip\doublerulesep}
\end{tabular}
\end{table*}

\textbf{}

\section{Preliminaries \label{sec:Preliminaries}}

Let $G=(V,E)$ be an undirected, unweighted graph where $\left|V\right|=n$
and $\left|E\right|=m$. For simplicity, we take $V=\left\{ 1,\dots,n\right\} $.
For a set $S\subseteq V$, let $E(S)$ be the set of edges in the
graph induced by $S$. For a node $u\in V$, let $\deg u$ be the
number of neighbors of $u$. We let $\Delta=\max_{u\in V}\deg u$,
i.e the maximum degree of a node in $V$. We use $[k]$ to denote
the set of integers from $1$ to $k$, and $\opt$ to denote the maximum
density of a subgraph. 

\textbf{Charikar's LP for DSG}$\quad$The LP for finding a densest
subgraph was introduced by \citet{charikar2000greedy} as follows
\begin{align}
\max_{x\ge0} & \sum_{e=uv\in E}\min\left\{ x_{u},x_{v}\right\} \ \mbox{st. }\sum_{u\in V}x_{u}\le1.\label{eq:primal}
\end{align}
\citet{charikar2000greedy} showed that given a feasible solution
$x$ to LP (\ref{eq:primal}) with objective $D$, we can construct
a set $S\subseteq V$ such that the density of $S$ is at least $D$.
The construction takes $O(n\log n+m)$ time: first, sort $\left(x_{v}\right)_{v\in V}$
in a decreasing order then select the prefix set $S$ that maximizes
$\frac{\left|E(S)\right|}{\left|S\right|}$. For this reason, we can
find a $\left(1-\epsilon\right)$ approximately densest subgraph by
finding a $\left(1-\epsilon\right)$ approximate solution to (\ref{eq:primal}).

\textbf{Dual LP}$\quad$\textbf{ }The dual of LP (\ref{eq:primal})
can be written as follows
\begin{align}
\min_{D,z\ge0}D\mbox{ st. } & \sum_{e\in E,u\in e}z_{eu}\le D,\quad\forall u\in V\label{eq:dual}\\
 & \ z_{eu}+z_{ev}\ge1,\quad\forall e=uv\in E.\nonumber 
\end{align}

\textbf{Width-reduced dual LP}$\quad$We use the width reduction technique
introduced in \citet{bahmani2014efficient} to improve the guaranteed
runtime. Since there is always an optimal solution $z$ for the dual
that satisfies $z\le q$ for $q\ge1$, adding this explicit constraint
to the LP as in (\ref{eq:dual-wr}) does not change the objective
of the optimal solutions.
\begin{align}
\min D\mbox{ st. } & \sum_{e\in E,u\in e}z_{eu}\le D,\quad\forall u\in V\label{eq:dual-wr}\\
 & \ z_{eu}+z_{ev}\ge1,\quad\forall e=uv\in E\nonumber \\
 & \ 0\le z_{eu}\le q,\quad\forall e,u\in e\in E.\nonumber 
\end{align}
By parameterizing $D$, \citet{bahmani2014efficient} showed that
we can solve the feasibility version of LP (\ref{eq:dual-wr}) in
$O\left(\frac{m\log m}{\epsilon^{2}}\right)$ time and achieve the
same total time via binary search for the optimal objective. However,
the downside of using the with-reduced LP (\ref{eq:dual-wr}) is that
it corresponds to a different primal than the LP (\ref{eq:primal}).
Thus it is not immediate how one can find an integral solution to
the DSG problem from a solution to (\ref{eq:dual-wr}). Note that,
\citet{bahmani2014efficient} used $q=2$---that is $0\le z_{eu}\le2,\quad\forall e,u\in e\in E$,
which is different from the natural choice of $q=1$. This value of
$q>1$ plays an important role in their intricate rounding scheme,
which involves discretization of the solution and a line sweep. In
contrast, we will show an algorithm that solves (\ref{eq:dual-wr})
for $q=1$ and also retains the simple rounding procedure by \citet{charikar2000greedy}.
Henceforth, we will refer to (\ref{eq:dual-wr}) with $q=1$.

\textbf{Dense subgraph decomposition and quadratic program}$\quad$The
dense subgraph decomposition problem \citep{tatti2015density} extends
DSG in that the output is a partition $S_{1}\cup\dots\cup S_{k}$
of the graph, where for $i\ge1$, $S_{i}$ is the maximal set that
maximizes $\left|E\left(\cup_{j=1}^{i-1}S_{j}\cup S\right)-E\left(\cup_{j=1}^{i-1}S_{j}\right)\right|/\left|S\right|$.
By this, one can simply recover the maximal densest subgraph by outputting
$S_{1}$. \citet{harb2022faster,harb2023convergence} showed that
this problem can be solved via the following quadratic program
\begin{align}
\min\sum_{u\in V}b_{u}^{2}\mbox{ st. } & b_{u}=\sum_{e\in E,u\in e}z_{eu},\quad\forall u\in V\label{eq:dual-quadratic}\\
 & z_{eu}+z_{ev}\ge1,\quad\forall e=uv\in E\nonumber \\
 & 0\le z_{eu}\le1,\quad\forall e,u\in e\in E.\nonumber 
\end{align}
\citet{harb2022faster} also showed that there exists a unique optimal
solution $b^{*}$ to (\ref{eq:dual-quadratic}). More precisely, for
the dense decomposition $S_{1}\cup\dots\cup S_{k}$, and $u\in S_{i}$,
we have $b_{u}^{*}=\frac{\left|E\left(\cup_{j=1}^{i}S_{j}\right)-E\left(\cup_{j=1}^{i-1}S_{j}\right)\right|}{\left|S_{i}\right|}.$
One can solve (\ref{eq:dual-quadratic}) by convex optimization tools
such as Frank-Wolfe algorithm \citep{danisch2017large,harb2023convergence}
or the accelerated FISTA algorithm \citep{harb2022faster,beck2009fast}.
\citet{harb2022faster} also introduced a rounding scheme called fractional
peeling to obtain an approximately densest subgraph decomposition
(see definition \ref{def:Dense-decomposition}).

\section{Algorithm via Multiplicative Weights Update\label{sec:Multiplicative-weights-update}}

In this section we present our algorithm to find a $(1-\epsilon)$
approximate solution to LP (\ref{eq:primal}). First, we give an overview
of our approach. The approach falls into the framework of MWU \citep{arora2012multiplicative}.
Instead of directly working with the dual LP (\ref{eq:dual}), we
will work with the width-reduced LP (\ref{eq:dual-wr}) with $q=1$.
We introduce dual variables $p\in\Delta_{m}$ which correspond to
the constraints $z_{eu}+z_{ev}\ge1$ for $e=uv\in E$. In each iteration
of the algorithm, given the values of $p$, we maintain a solution
$z$ that satisfies the combined constraint $\sum_{e\in E}p_{e}\left(z_{eu}+z_{ev}\right)\ge1$.
Note that, we can always make equality happens without increasing
the objective. This reduces to solving the following LP
\begin{align}
\min_{z\in[0,1]^{2m}}\max_{u\in V} & \ \sum_{e\in E,u\in e}z_{eu}\label{eq:dual-wr-iteration}\\
\sum_{e\in E}p_{e}\left(z_{eu}+z_{ev}\right) & =1\label{eq:2}
\end{align}
The average solution for $z$ ensures that the constraints of LP (\ref{eq:dual-wr})
are satisfied approximately. To update $p$, we use MWU. Using the
value of $p$ in the best iteration, we can construct a feasible solution
to the primal LP (\ref{eq:primal}) with objective at least $\left(1-\epsilon\right)\opt$.
This solution allows us to use Charikar's rounding procedure to obtain
an approximately densest subgraph (see Section \ref{sec:Preliminaries}).

What differs from \citep{bahmani2014efficient} is that we directly
solve problem (\ref{eq:dual-wr-iteration}) instead of parametrizing
$D=\max_{u\in V}\sum_{e\in E,u\in e}z_{eu}$ and solving the feasibility
version of (\ref{eq:dual-wr-iteration}). There are two reasons why
this is a better approach. First, being able to exactly optimize LP
(\ref{eq:dual-wr-iteration}) allows us to use complementary slackness
and recover the primal solution in a simple way. We show that this
primal solution satisfies Charikar's LP (\ref{eq:primal}) and thus
allows us to use Charikar's simple rounding procedure. In this way,
we completely remove the involved rounding procedure in \citet{bahmani2014efficient}.
Second, there is no longer need for using binary search which could
be a concern for the runtime in practice.

\subsection{Algorithm for solving problem (\ref{eq:dual-wr}) \label{subsec:Obtaining-z}}

\begin{algorithm}
{\small{}\caption{{\small{}Multiplicative Weights Update for solving (\ref{eq:dual-wr})}}
\label{alg:mwu}}{\small\par}

{\small{}Let $T=\frac{2\ln m}{\epsilon^{2}}$, $\eta=\epsilon$}{\small\par}

{\small{}Initialize $p^{(1)}=\left(\frac{1}{m},\dots,\frac{1}{m}\right)$,
$G^{(0)}=0\in\R^{m}$}{\small\par}

{\small{}for $t=1\dots T$}{\small\par}

{\small{}$\quad$Let $z^{(t)}$ be an optimal solution to (\ref{eq:dual-wr-iteration})
for $p=p^{(t)}$}{\small\par}

{\small{}$\quad$Let $g_{e}^{(t)}=1-\left(z_{eu}^{(t)}+z_{ev}^{(t)}\right)$
for all $e\in E$}{\small\par}

{\small{}$\quad$Let $G^{(t)}=\sum_{\tau=1}^{t}g^{(\tau)}$}{\small\par}

{\small{}$\quad$Let $p^{(t+1)}=\nabla\smax_{\eta}(G^{(t)})$, ie,
$p_{e}^{(t+1)}=\text{\ensuremath{\frac{\exp(\eta G_{e}^{(t)})}{\sum_{e'}\exp(\eta G_{e'}^{(t)})}}}$}{\small\par}

{\small{}Output $\frac{1}{T}\sum_{t=1}^{T}z^{(t)}$}{\small\par}
\end{algorithm}

In this section, we give our algorithm solving LP (\ref{eq:dual-wr}),
shown in Algorithm \ref{alg:mwu}. The algorithm is based on the multiplicative
weights framework and it uses as a subroutine an algorithm that, given
$p\in\Delta_{m}$, it returns an optimal solution $z$ to the LP (\ref{eq:dual-wr-iteration}).
We show how to efficiently implement this subroutine in the next section.
The following lemma and its corollary show that the output of Algorithm
\ref{alg:mwu} is approximately optimal for (\ref{eq:dual-wr}).
\begin{lem}
\label{lem:mwu-average-solution}Let $z^{*}$ be an optimal solution
to LP (\ref{eq:dual-wr}). Algorithm \ref{alg:mwu} outputs $\overline{z}=\frac{1}{T}\sum_{t=1}^{T}z^{(t)}$
that satisfies
\begin{align*}
\max_{u\in V}\sum_{e\in E,u\in e}\overline{z}_{eu} & \le\max_{u\in V}\sum_{e\in E,u\in e}z_{eu}^{*}
\end{align*}
and for all $e=uv\in E$
\begin{align*}
\overline{z}_{eu}+\overline{z}_{ev} & \ge1-\epsilon.
\end{align*}
\end{lem}

\begin{cor}
\label{cor:existence-good-D}Let $D^{(t)}=\max_{u\in V}\sum_{e\in E,u\in e}z_{eu}^{(t)}$
and $D^{*}=\max_{u\in V}\sum_{e\in E,u\in e}z_{eu}^{*}$. There is
$t\in[T]$ such that 
\begin{align*}
D^{(t)} & \ge\left(1-\epsilon\right)D^{*}.
\end{align*}
\end{cor}

\subsection{Algorithm for solving problem (\ref{eq:dual-wr-iteration})\label{subsec:Solving}}

In this section, we give an efficient algorithm that, given $p\in\Delta_{m}$,
it returns an optimal solution $z$ to the LP (\ref{eq:dual-wr-iteration}).
We write the constraint (\ref{eq:2}) of the LP as $\sum_{u\in V}\sum_{e\in E,u\in e}p_{e}z_{eu}=1$.
The intuition to solve LP (\ref{eq:dual-wr-iteration}) follows from
\citet{bahmani2014efficient}: given a guess $D$ for the optimal
objective, we can now think of LP (\ref{eq:dual-wr-iteration}) as
solving a feasibility knapsack problem, for which the strategy is
greedily packing the items, i.e, setting $z_{eu}=1$, in the decreasing
order of $p_{e}$.

Returning to LP (\ref{eq:dual-wr-iteration}), we proceed by first
sorting for each $u$ all of the edges incident to $u$ in the decreasing
order of $p_{e}$. For two edges $e$ and $e'$ incident to $u$,
we write $e\prec_{u}e'$ if $e$ precedes $e'$ in this order. We
show the following lemma
\begin{lem}
\label{lem:assignment}Let $D^{*}$ be the optimal objective of LP
(\ref{eq:dual-wr-iteration}). Let $z^{*}$ be such that $z_{eu}^{*}=\min\left\{ 1,D^{*}-\sum_{e'\prec_{u}e\colon u\in e'}z_{e'u}\right\} $.
Then $z^{*}$ is an optimal solution to LP (\ref{eq:dual-wr-iteration}).
\end{lem}

We consider $D$ as a variable we need to solve for and assign value
of $z$ according to Lemma \ref{lem:assignment}. In this way for
any value $D\in\left[0,\Delta\right]$, for each $u$, the first $\min\left\{ \lfloor D\rfloor,\deg u\right\} $
edges in the decreasing order of $p$ incident to $u$ have $z_{eu}=1$,
the next edge (if exists) has value $z_{eu}=R\coloneqq D-\lfloor D\rfloor$
and the remaining edges have value $z_{eu}=0$. Also note that for
a solution, we have 
\begin{align}
\sum_{u\in V}\sum_{e\in E,u\in e}p_{e}z_{eu} & =1.\label{eq:1}
\end{align}
Thus we can proceed by testing all values of $\lfloor D\rfloor\in\left[\Delta\right]$.
For each value of $\lfloor D\rfloor$, $R$ is determined by solving
Equation (\ref{eq:1}). We choose the smallest $\lfloor D\rfloor$
such that $0\le R<1$. 

We summarize this procedure in Algorithm \ref{alg:iteration-solver}.

\begin{algorithm}
{\small{}\caption{{\small{}Solver for (\ref{eq:dual-wr-iteration})}}
\label{alg:iteration-solver}}{\small\par}

\textbf{\small{}Input:}{\small{} $p\in\Delta_{m}$}{\small\par}

{\small{}For each $u\in V$, the edges incident to $u$ in non-increasing
order according to $p_{e}$}{\small\par}

{\small{}for $\lfloor D\rfloor\in[0,\Delta]:$}{\small\par}

{\small{}$\quad$for $u\in V$, let $z_{eu}=\min\{1,\lfloor D\rfloor-\sum_{e'\prec_{u}e\colon u\in e'}z_{e'u}\}$.
Let $E(u)$ be the set of $e$ incident to $u$ such that $z_{eu}=1$
and $\overline{E}(u)$ be the remaining edges.}{\small\par}

{\small{}$\quad$Let $p(u)=\max\left\{ p_{e}:e\in\overline{E}(u)\right\} $
(or $0$ if $\overline{E}(u)=\emptyset$)}{\small\par}

{\small{}$\quad$Let $R=\frac{1-\sum_{u\in V}\sum_{e\in E(u)}p_{e}}{\sum_{u\in V}p(u)}$}{\small\par}

{\small{}$\quad$if $0\le R<1:$}{\small\par}

{\small{}$\quad\quad$return $\lfloor D\rfloor+R$, $z$}{\small\par}
\end{algorithm}

\begin{lem}
Algorithm \ref{alg:iteration-solver} outputs an optimal solution
for LP (\ref{eq:dual-wr-iteration}) in time $O(m\log\Delta)$.
\end{lem}

\begin{proof}
The correctness of the algorithm is ensured by the fact that we output
the first (smallest) $D$ that gives an assignment according to Lemma
\ref{lem:assignment}. Sorting the edges for each node $u$ takes
$O\left(\deg u\log\deg u\right)$ time, hence the total sorting time
is $O(m\log\Delta)$. The assignment of $z$ also takes at most $O\left(m\right)$
since during the course of the algorithm, each $z_{eu}$ is used for
computing the value of $R$ at most once. Therefore the total runtime
is $O(m\log\Delta)$.
\end{proof}

\subsection{Constructing the solution\label{subsec:Obtaining-x} }

Finally, we show a way to construct a solution to the primal LP (\ref{eq:primal}).
Let $\tau$ be the iteration $t$ that has the biggest value of $D^{(t)}=\max_{v\in V}\sum_{e\in E,u\in e}z_{eu}^{(t)}$.
From Corollary \ref{cor:existence-good-D}, we have $D^{(\tau)}\ge(1-\epsilon)D^{*}=(1-\epsilon)\opt$.
The primal program corresponding to the dual LP (\ref{eq:dual-wr-iteration})
is as follows
\begin{align}
\max_{x,\alpha\ge0,W} & \ W-\sum_{e=uv}\left(\alpha_{eu}+\alpha_{ev}\right)\label{eq:primal-combined-constraint}\\
p_{e}^{(\tau)}W & \leq\min\{x_{u}+\alpha_{eu},x_{v}+\alpha_{ev}\}\quad\forall e=uv\nonumber \\
\sum_{v}x_{v} & \leq1.\nonumber 
\end{align}
Recall that the solution $z^{(\tau)}$ of (\ref{eq:dual-wr-iteration})
is obtained as follows. We sort the edges in the decreasing order
according to $p_{e}^{(\tau)}$. When considering $e=uv$, we set $z_{eu}^{(\tau)}=\min\left\{ 1,D^{(\tau)}-\sum_{e'\prec e\colon u\in e'}z_{e'u}^{(\tau)}\right\} $.
Let $X=\left\{ u:\sum_{e}z_{eu}^{(\tau)}=D^{(\tau)}\right\} $. For
$u\in X$, let $e(u)$ be the edge with smallest $p_{e}^{(\tau)}$
among the edges with $z_{eu}^{(\tau)}>0$, let $W=\frac{1}{\sum_{u\in X}p_{e(u)}^{(\tau)}}$
. Set
\begin{align*}
x_{u} & =p_{e(u)}^{(\tau)}W\\
\alpha_{eu} & =p_{e}^{(\tau)}W-x_{u}\ge0\quad\forall e:p_{e}^{(\tau)}\ge p_{e(u)}^{(\tau)}\\
\alpha_{eu} & =0\quad\forall e:p_{e}^{(\tau)}<p_{e(u)}^{(\tau)}
\end{align*}
For $u\notin X$, set $x_{u}=0;\alpha_{eu}=p_{e}^{(\tau)}W\quad\forall e\ni u$.
We can verify that $\left(W,x,\alpha\right)$ is an optimal solution
to LP (\ref{eq:primal-combined-constraint}) by complementary slackness
and $x$ is an $(1-\epsilon)$-approximate solution to (\ref{eq:primal})
by strong duality.
\begin{lem}
\label{lem:mwu-complementary-slackness}$\left(W,x,\alpha\right)$
is an optimal solution to LP (\ref{eq:primal-combined-constraint}).
\end{lem}

\begin{lem}
\label{lem:primal-optimal}$x$ is a $(1-\epsilon)$-approximate solution
to (\ref{eq:primal}).
\end{lem}

\begin{rem}
\label{rem:rounding}As we can see here, the new insight is that,
more generally, as long as we have $p\in\Delta_{m}$ for which we
know that the objective of the LP (\ref{eq:dual-wr-iteration}) is
at least $D$, we can obtain a subgraph with density at least $D$. 
\end{rem}

\subsection{Final runtime}

Combining subroutines from Section \ref{subsec:Obtaining-z}-\ref{subsec:Obtaining-x},
we obtain the following result.
\begin{thm}
There exists an algorithm that outputs a subgraph of density $\ge(1-\epsilon)\opt$
in $O\left(\frac{\log m}{\epsilon^{2}}\right)$ iterations, each of
which can be implemented in $O\left(m\log\Delta\right)$ time for
a total $O\left(\frac{m\log\Delta\log m}{\epsilon^{2}}\right)$ time.
\end{thm}

\section{Algorithm via Area Convexity\label{sec:area-convex}}

In this section, by building on the approaches based on area convexity
\citep{sherman2017area,boob2019faster}, we obtain an algorithm with
an improved iteration complexity of $O(\frac{\log m}{\epsilon})$,
and the same nearly-linear time per iteration. Our algorithm improves
upon the result of \citet{boob2019faster} by a factor $\Delta$ (the
maximum degree in the graph), in both the number of iterations and
overall running time. This improvement comes from the following reasons
where we depart from \citet{boob2019faster}. First, taking inspiration
from \citet{bahmani2014efficient} and the width reduction technique,
we parametrize $D=\max_{u\in V}\sum_{e\ni u}z_{eu}$, but keeping
the constraints $\forall u,\sum_{e\ni u}z_{eu}\le D$ as the domain
of $z$ instead of the constraints in the feasibility LP. This reduces
the width of the LP from $\Delta$ to $2$. Now the question becomes
whether we can solve each subproblem (ie., implement the oracle) efficiently.
To do so, we replace the area convex regularizer of \citet{boob2019faster}
with the choice in \citet{sherman2017area} and subsequent works \citep{jambulapati2019direct,jambulapati2023revisiting}.
This choice simplifies the regularizer to a quadratic function (with
respect to $z$) and allows to optimize for each vertex separately.
We also take inspiration from the oracle implementation for MWU (Algorithm
\ref{alg:iteration-solver}) and show that we can implement the oracle
in this case in $\tilde{O}(m)$ time. Finally, we show that the rounding
procedure used for Algorithm \ref{alg:mwu} can also be used to obtain
an integral solution, which was not known in \citet{boob2019faster}.

\subsection{Reduction to saddle point optimization\label{subsec:Reduction-to-saddle}}

First, we show that we can solve LP (\ref{eq:dual-wr}) for $q=1$
via a reduction to a saddle point problem. By parameterizing the variable
$D$, we convert solving LP (\ref{eq:dual-wr}) to the following feasibility
LP

\begin{align}
\exists?z\in C(D) & \mbox{ st. }z_{eu}+z_{ev}\ge1,\quad\forall e=uv\in E\label{eq:dual-wr-feasibility}\\
\mbox{where }C(D) & =\left\{ z\in\left[0,1\right]^{2m}:\forall u,\sum_{e\ni u}z_{eu}\le D\right\} .\nonumber 
\end{align}
For simplicity, we write the domain as $C$ when it is clear what
value $D$ is being used. Let us also denote the constraint matrix
by $B\in\R^{m\times2m}$ and let $A\coloneqq\left[\begin{array}{cc}
0 & B^{T}\\
-B & 0
\end{array}\right]$. In Lemma \ref{lem:ac-reduction} (from \citep{boob2019faster}),
we show that this feasibility problem can be reformulated as the following
saddle point problem
\begin{align}
\min_{z\in C,y\in\Delta_{m}}\max_{\overline{z}\in C,\overline{y}\in\Delta_{m}} & \sum_{e}y_{e}\left(\overline{z}_{eu}+\overline{z}_{ev}\right)-\overline{y}_{e}\left(z_{eu}+z_{ev}\right)\nonumber \\
=y^{T}B\overline{z}-\overline{y}^{T}Bz & =\left[\begin{array}{cc}
\overline{z}^{T} & \overline{y}^{T}\end{array}\right]A\left[\begin{array}{c}
z\\
y
\end{array}\right].\label{eq:saddle-point}
\end{align}
By approximately solving (\ref{eq:saddle-point}), we obtain $(z,y)$
such that either $z$ is an approximate solution to Problem (\ref{eq:dual-wr-feasibility})
or $y$ can certify that Problem (\ref{eq:dual-wr-feasibility}) is
infeasible.

\subsection{Algorithm for solving problem (\ref{eq:saddle-point})}

\label{subsec:saddle-point}

\begin{algorithm}
{\small{}\caption{{\small{}Solver for (\ref{eq:saddle-point}) using oracle $\Phi$
(Algorithm }\ref{alg:oracle}{\small{})}}
\label{alg:area-convexity}}{\small\par}

{\small{}Initialize $w^{(0)}=(z^{(0)},y^{(0)})\in C\times\Delta_{m}$
where $z^{(0)}=0$ and $y^{(0)}=\frac{1}{m}$}{\small\par}

{\small{}for $t=0,\dots,T-1$}{\small\par}

{\small{}$\quad$$w^{(t+1)}=w^{(t)}+\tilde{\Phi}\left(Aw^{(t)}\right)$
where $\tilde{\Phi}(a)=\Phi(a+2A\Phi(a))$.}{\small\par}
\end{algorithm}
\begin{algorithm}
{\small{}\caption{{\small{}Algorithm for oracle $\Phi$ (Definition }\ref{def:av-approximate-oracle}{\small{})}}
\label{alg:oracle}}{\small\par}

\textbf{\small{}Input:}{\small{} $x=(s,r)$, $s\in\R^{2m}$, $r\in\R^{m}$}{\small\par}

{\small{}Initialize $z^{(0)}=0$}{\small\par}

{\small{}Let $H(z,y)=\phi(z,y)-\left\langle z,s\right\rangle -\left\langle y,r\right\rangle $}{\small\par}

{\small{}for $t=0,\dots,T$}{\small\par}

{\small{}$\quad$$y^{(t+1)}=\arg\min_{y\in\Delta_{m}}H(z^{(t)},y)$}{\small\par}

{\small{}$\quad$$z^{(t+1)}=\arg\min_{z\in C}H(z,y^{(t+1)})$}{\small\par}

{\small{}return $\left(z^{(T+1)},y^{(T+1)}\right)$}{\small\par}
\end{algorithm}

Next, we describe the algorithm via the general area convexity technique
by \citet{sherman2017area} for solving problem (\ref{eq:saddle-point}).
In order to use this technique, one key point is to choose a regularizer
function which is area convex with respect to $A$ and has a small
range (width). The following regularizer function enjoys these properties
\begin{align}
\phi(z,y) & =6\sqrt{3}\left(\sum_{e\in E}y_{e}\left(z_{eu}^{2}+z_{ev}^{2}\right)+6y_{e}\log y_{e}-2\right).\label{eq:phi}
\end{align}
Let us now assume access to a $\delta$-approximate minimization oracle
$\Phi$ for solving subproblems regularized by $\phi$ in the following
sense.
\begin{defn}
\label{def:av-approximate-oracle}\citep{sherman2017area} A $\delta$-approximate
minimization oracle $\Phi$ for $\phi$ takes input $x\in\R^{3m}$
and output $w^{*}\in C\times\Delta_{m}$ such that 
\begin{align*}
\left\langle w^{*},x\right\rangle -\phi(w^{*})+\delta & \ge\sup_{w\in C\times\Delta_{m}}\left\langle w,x\right\rangle -\phi(w)\coloneqq\phi^{*}(x).
\end{align*}
\end{defn}

Once we have this oracle, we can use Sherman's algorithm (Algorithm
\ref{alg:area-convexity}) to approximately solve problem (\ref{eq:saddle-point}).
The convergence guarantee is given in Lemma \ref{lem:av-convergence}.

\begin{lem}
\label{lem:av-convergence}For the choice of $\phi$ in (\ref{eq:phi}),
Algorithm \ref{alg:area-convexity} outputs $w_{T}$ that satisfies
$\frac{w^{(T)}}{T}\in C\times\Delta_{m}$ and
\begin{align*}
\sup_{\overline{w}\in C\times\Delta_{m}}\overline{w}A\frac{w^{(T)}}{T} & \le\delta+O\left(\frac{\log m}{T}\right).
\end{align*}
\end{lem}

\textbf{Oracle implementation.}$\quad$We now show that the oracle
can be implemented efficiently via alternating minimization (Algorithm
\ref{alg:oracle}). We show that Algorithm \ref{alg:oracle} enjoys
linear convergence and can be implemented efficiently in the following
Lemmas.
\begin{lem}
\label{lem:alternating-min-lemma}Let $\left(z_{\opt},y_{\opt}\right)\in\arg\min_{(z,y)\in C\times\Delta_{m}}H(z,y)$.
For $T=O\left(\log\frac{\left(H(z^{(0)},y^{(1)})-H\left(z_{\opt},y_{\opt}\right)\right)}{\delta}\right)$,
$\left(z^{(T+1)},y^{(T+1)}\right)$ satisfies
\begin{align*}
H(z^{(T+1)},y^{(T+1)})-H\left(z_{\opt},y_{\opt}\right) & \le\delta.
\end{align*}
\end{lem}

\begin{lem}
\label{lem:av-iteration-time}Each iteration of Algorithm \ref{alg:oracle}
can be implemented in $O(m\log\Delta).$
\end{lem}

\subsection{Constructing the solution}

Putting together the reduction from Section \ref{subsec:Reduction-to-saddle}
and the algorithm from Section \ref{subsec:saddle-point}, we obtain
an algorithm that returns an approximate solution $z$ to the feasibility
Problem (\ref{eq:dual-wr-feasibility}) or it returns that (\ref{eq:dual-wr-feasibility})
is infeasible. By combining this algorithm with binary search over
$D$, we obtain the result in the following theorem. We note that
we can use the binary search approach of \citet{bahmani2014efficient}
to avoid incurring any extra overhead in the running time.
\begin{thm}
There exists an algorithm that outputs $z$ and $\tilde{D}=\max\sum_{e\ni u}z_{eu}$
such that $z_{eu}+z_{ev}\ge1$ for all $e=uv\in E$ and where $D^{*}(1-\epsilon)\le\tilde{D}\le D^{*}(1+\epsilon)$
and $D^{*}=\opt$ is the optimal value of $D$ in LP (\ref{eq:dual-wr}).
\end{thm}

Finally, to reconstruct the integral solution, let $\overline{D}=\tilde{D}(1-2\epsilon)$
and $(\overline{z},\overline{y})$ be an $\epsilon$-approximate solution
for problem (\ref{eq:saddle-point}) on domain $C(\overline{D})$
output by Algorithm \ref{alg:area-convexity}. We show the following
lemma:
\begin{lem}
\label{lem:av-objective}The objective of the following LP
\begin{align*}
\min_{z\in[0,1]^{2m}}D\mbox{ st. } & \sum_{e\ni u}z_{eu}\le D;\ \sum_{e\in E}\overline{y}_{e}\left(z_{eu}+z_{ev}\right)=1.
\end{align*}
is strictly more than $\overline{D}>(1-3\epsilon)\opt$.
\end{lem}

Due to this lemma and Remark \ref{rem:rounding}, we can now follow
the procedure in Section \ref{subsec:Obtaining-x} and reconstruct
the primal solution and obtain an $(1-3\epsilon)$-approximately densest
subgraph.

\subsection{Final runtime}
\begin{thm}
There exists an algorithm that outputs a subgraph of density $\ge(1-\epsilon)\opt$
in $O\left(\frac{\log m}{\epsilon}\right)$ iterations, each of which
can be implemented in $O\left(m\log\Delta\log\frac{1}{\epsilon}\right)$
time for a total $O\left(\frac{m}{\epsilon}\log m\log\Delta\log\frac{1}{\epsilon}\right)$
time.
\end{thm}

\section{Algorithm via Random Coordinate Descent\label{sec:DSM}}

In this section, we give an algorithm for finding an approximate dense
decomposition. First, we recall the definition of an $\epsilon$-approximate
dense decomposition.
\begin{defn}
\label{def:Dense-decomposition} \citep{harb2022faster} We say a
partition $T_{1},\dots,T_{r}$ is an $\epsilon$-approximate dense
decomposition to $S_{1},\dots,S_{k}$ (the true decomposition) if,
for all $i,j$ and $S_{i}\cap T_{j}\neq\emptyset$ then $\frac{\left|E(T_{j})\right|+\left|E\left(T_{j},\cup_{h<j}T_{h}\right)\right|}{\left|T_{j}\right|}\ge\frac{\left|E(S_{i})\right|+\left|E\left(S_{i},\cup_{h<j}S_{h}\right)\right|}{\left|S_{i}\right|}-\epsilon.$
\end{defn}

We adapt the accelerated random coordinate descent of \citet{ene2015random}
to find a fractional solution to (\ref{eq:dual-quadratic}) and then
use the fractional peeling procedure by \citet{harb2022faster} to
obtain the decomposition. 

\subsection{Continuous formulation}

We first write (\ref{eq:dual-quadratic}) in the following equivalent
way. For each $e\in E$, let $F_{e}:2^{V}\to\R$ be such that $F_{e}(S)=1$
if $e\subseteq S$, $F_{e}(S)=0$ otherwise. We have $F_{e}$ is a
supermodular function since $F_{e}(S)+F_{e}(T)\le F_{e}(S\cup T)+F_{e}(S\cap T)$
for all $S,T\subseteq V$. The base contrapolymatroid of $F_{e}$
is
\begin{align*}
B(F_{e}) & =\left\{ z_{e}\in\R^{n}\colon z_{e}(S)\ge F_{e}(S)\;\forall S\subseteq V,z_{e}(V)=1\right\} 
\end{align*}
We show in Appendix \ref{sec:appdx-DSM} that (\ref{eq:dual-quadratic})
is equivalent to
\begin{align}
\min_{z:z_{e}\in B(F_{e}),\forall e\in E} & f(z)\coloneqq\left\Vert \sum_{e\in E}z_{e}\right\Vert _{2}^{2}\label{eq:supermod-min}
\end{align}
Problem (\ref{eq:supermod-min}) has exactly the same form as the
continuous formulation for decomposable submodular minimization studied
in \citet{nishihara2014convergence,ene2015random}, except that now
we minimize over the base contrapolymatroid of a supermodular function
instead of the base polytope of a submodular function. We provide
further details about this connection in Appendix \ref{sec:appdx-DSM}.
This connection allows us to adapt the Accelerated Coordinate Descent
algorithm by \citet{ene2015random} to solve (\ref{eq:dual-quadratic}).

\subsection{Accelerated Random Coordinate Descent }

For an edge $e=uv\in E$, \citet{harb2022faster} show that projection
onto $B(F_{e})$ can be done via the following operator $\proj_{e}$.
For simplicity, we only consider the relevant component $s_{eu}$
and $s_{ev}$ of $s$ (the remaining components are all $0$). The
projected solution onto $B(F_{e})$, $\proj_{e}((s_{eu}\ s_{ev}))$,
is given by
\begin{align*}
\proj_{e}((s_{eu}\ s_{ev})) & =\begin{cases}
\left(\frac{s_{eu}-s_{ev}+1}{2}\ \frac{s_{ev}-s_{eu}+1}{2}\right) & \mbox{if }\left|s_{eu}-s_{ev}\right|\le1\\
(1\ 0) & \mbox{if }s_{eu}-s_{ev}>1\\
(0\ 1) & \mbox{otherwise}.
\end{cases}
\end{align*}
Note that with $\proj_{e}$, for $x,y\in\R^{n}$ and $\eta>0$, we
can solve the following problem in $O(1)$ time.
\begin{align*}
\arg\min_{s\in B(F_{e})}\bigg(\left\langle \nabla_{e}f(x),(s_{eu}\ s_{ev})\right\rangle +\eta\left\Vert s-y\right\Vert _{2}^{2}\bigg)= & \proj_{e}\bigg((y_{eu}\ y_{ev})-\frac{1}{2\eta}\nabla_{e}f(x)\bigg).
\end{align*}
where we use $\nabla_{e}f\in\R^{2}$ to denote the gradient with respect
to the component $eu$ and $ev$. We present the Accelerated Random
Coordinate Descent Algorithm in Algorithm \ref{alg:acdm}. The algorithm
and its convergence analysis stay close to the analysis in \citet{ene2015random},
which we omit.

\begin{algorithm}[H]
{\small{}\caption{{\small{}Accelerated Random Coordinate Descent}}
\label{alg:acdm}}{\small\par}

{\small{}Initialize $z^{(0)}\in{\cal P}$}{\small\par}

{\small{}for $k=1\dots K=O\left(\log\frac{n}{\epsilon}\right)$:}{\small\par}

{\small{}$\quad$$y^{(k,0)}=z^{(k-1)}\in{\cal P}$, $\theta^{(k,0)}=\frac{1}{m},$
$w^{(k,0)}=0$}{\small\par}

{\small{}$\quad$for $t=1\dots T=O\left(mn\right)$:}{\small\par}

{\small{}$\quad\quad$select a set $R^{(t)}$ of edges, each $e\in E$
with probability $\frac{1}{m}$}{\small\par}

{\small{}$\quad\quad$for $e\in R^{(t)}:$}{\small\par}

{\small{}$\quad\quad\quad$$x^{(k,t)}=\theta^{(k,t-1)2}w^{(k,t-1)}+y^{(k,t-1)}$}{\small\par}

{\small{}$\quad\quad\quad$$y^{(k,t)}=\arg\min_{s\in B(F_{e})}\Bigg(\left\langle \nabla_{e}f(x^{(k,t)}),(s_{eu}\ s_{ev})\right\rangle +2m\theta^{(k,t-1)}\left\Vert (s_{eu}\ s_{ev})-(y_{eu}^{(k,t-1)}\ y_{ev}^{(k,t-1)})\right\Vert _{2}^{2}\Bigg)$}{\small\par}

{\small{}$\quad\quad\quad$$w^{(k,t)}=w^{(k,t-1)}-\frac{1-m\theta^{(k,t-1)}}{\theta^{(k,t-1)2}}\left(y^{(k,t)}-y^{(k,t-1)}\right)$}{\small\par}

{\small{}$\quad\quad\quad$$\theta^{(k,t)}=\frac{\sqrt{\theta^{(k,t-1)4}+4\theta^{(k,t-1)2}}-\theta^{(k,t-1)2}}{2}$}{\small\par}

{\small{}$\quad$$z^{(k)}=\theta^{(k,T-1)2}w^{(k,T)}+y^{(k,T)}$}{\small\par}

{\small{}return $z^{(K)}$}{\small\par}
\end{algorithm}

The runtime of Algorithm \ref{alg:acdm} is given in the next lemma.
\begin{lem}
\label{lem:quadratic-convergence-1}Algorithm \ref{alg:acdm} produces
in expected time $O\left(mn\log\frac{n}{\epsilon}\right)$ a solution
$z$ such that $\E\left[f(z)-f(z^{*})\right]\le\epsilon.$
\end{lem}

\subsection{Fractional peeling}

The fractional peeling procedure in \citet{harb2022faster} takes
an $\epsilon$ approximate solution $(z,b)$ for the program (\ref{eq:dual-quadratic})
in the sense that $\left\Vert b-b^{*}\right\Vert _{2}\le\epsilon$
and returns an $\epsilon\sqrt{n}$- approximate dense decomposition.
The procedure is as follows: 

For the first partition $T_{1}$, starting with $b'=b$ and $G^{(0)}=G$,
in each iteration $t$, we select the vertex $u$ with the smallest
value $b'_{u}$ and update $b'_{v}\gets b'_{v}-z_{ev}$ for all $v$
adjacent to $u$ in $G^{(t-1)}$ and $G^{(t)}\gets G^{(t-1)}-u$.
We return the graph $G^{(t)}$ with the maximum density. For the subsequent
partition $T_{t}$, we update for all $e=uv$ such that $u\in T_{1}\cup\dots\cup T_{t-1}$
and $v\in G\setminus(T_{1}\cup\dots\cup T_{t-1})$: $z_{eu}\gets0$
and $z_{ev}\gets1$. Remove $T_{1}\cup\dots\cup T_{t-1}$ from $G$
and repeat the above procedure for the remaining graph.

\citet{harb2022faster} show the following result.
\begin{lem}
For $(z,b)$ satisfying $\left\Vert b-b^{*}\right\Vert _{2}\le\epsilon$,
the fractional peeling procedure described above output $\epsilon\sqrt{n}$-
approximate dense decomposition in $\tilde{O}(mn)$ time.
\end{lem}

\subsection{Final runtime}

Combining the guarantees of Algorithm \ref{alg:acdm} and the fractional
peeling procedure, we obtain the following result.
\begin{thm}
Algorithm \ref{alg:acdm} and the fractional peeling procedure \citep{harb2022faster}
output an $\epsilon$-approximate dense decomposition in $O\left(mn\log\frac{n}{\epsilon}\right)$
time in expectation.
\end{thm}

\section{Experiments\label{sec:Experiments}}

\begin{figure}
{\small{}}\subfloat[Iteration/Best Density]{{\small{}\includegraphics[width=0.48\columnwidth]{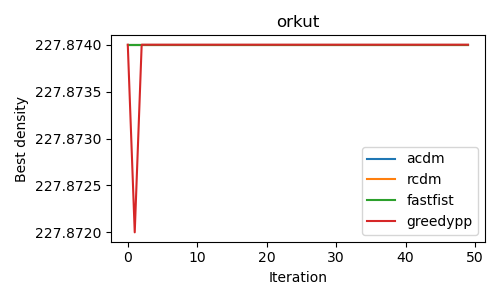}\includegraphics[width=0.48\columnwidth]{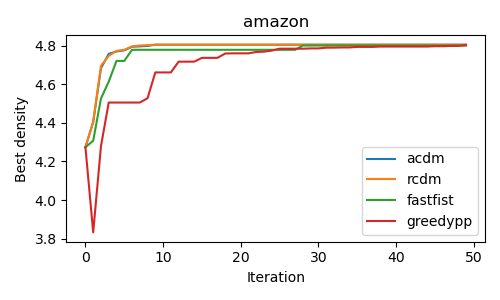}}{\small\par}

{\small{}\label{fig:iteration-best-density}\vspace{-10bp}
}{\small\par}}{\small{}\hfill{}\vspace{-10bp}
}\subfloat[Time/Best Density]{{\small{}\includegraphics[width=0.48\columnwidth]{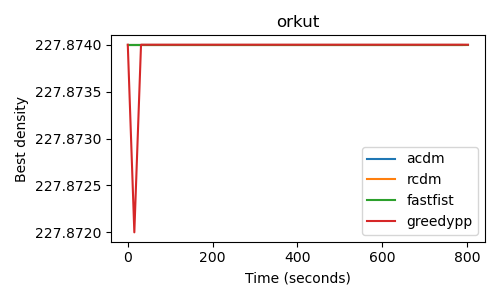}\includegraphics[width=0.48\columnwidth]{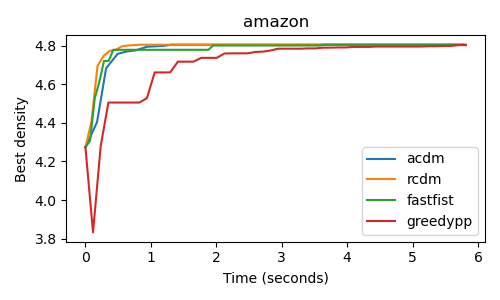}}{\small\par}

{\small{}\label{fig:time-best-density}\vspace{-10bp}
}{\small\par}}{\small{}\hfill{}\vspace{-10bp}
}\subfloat[Iteration/Load norm $\left(\sum_{u\in V}b_{u}^{2}\right)^{1/2}$]{{\small{}\includegraphics[width=0.48\columnwidth]{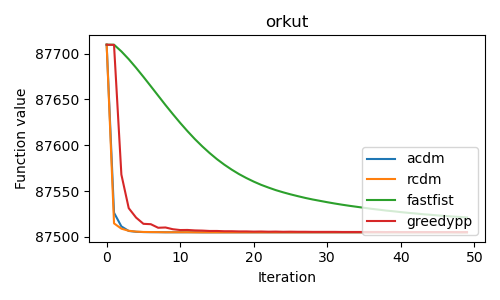}\includegraphics[width=0.48\columnwidth]{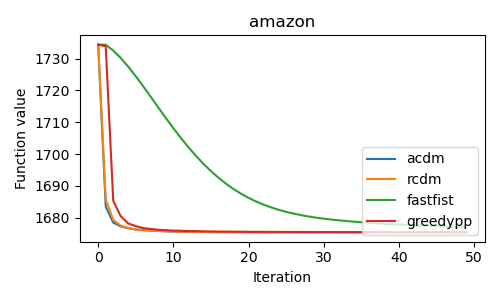}}{\small\par}

{\small{}\label{fig:iteration-function-value}\vspace{-10bp}
}{\small\par}}{\small\par}

{\small{}\caption{{\small{}Experiment results on orkut and com-Amazon. In the legends:
acdm, rcdm, fastfist, greedypp respectively represent Algorithm \ref{alg:acdm},
Algorithm \ref{alg:rcdm}, FISTA-based algorithm by \citet{harb2022faster}
and Greedy++ \citep{boob2020flowless}}}
}{\small\par}
\end{figure}

In this section, we compare the performance of existing algorithms
and Algorithm \ref{alg:acdm}. We also consider the version of Algorithm
\ref{alg:acdm} without acceleration, Algorithm \ref{alg:rcdm} shown
in the appendix. On the other hand, due to the involved subroutines,
we do not compare the performance of Algorithm \ref{alg:mwu} and
\ref{alg:area-convexity}. We follow the experimental set up in prior
works, including \citet{boob2020flowless} and \citet{harb2022faster}.

\textbf{Benchmark.$\quad$}We consider three algorithms: Frank-Wolfe
\citep{danisch2017large}, implemented by \citet{harb2022faster},
Greedy++ \citep{boob2020flowless} and FISTA for DSG \citep{harb2022faster}.

\textbf{Implementation.}$\quad$We use the implementation of all benchmark
algorithms provided by \citet{harb2022faster}. The implementation
of our algorithms also uses the same code base by \citet{harb2022faster}.
For practical purposes, we modify Algorithms \ref{alg:acdm} and \ref{alg:rcdm}
by replacing the inner loop with making passes over random permutations
of the edges instead of random sampling edges. We also restart Algorithm
\ref{alg:acdm} after each pass that increases the function value.
We show the pseudocode for these variants in Algorithms \ref{alg:acdm-practical}
in the appendix and in Option 2 of Algorithm \ref{alg:rcdm}. For
all algorithms, we initialize at the solution by the Greedy peeling
algorithm \citep{charikar2000greedy}.

\textbf{Datasets.$\quad$}The algorithms are compared on eight different
datasets, summarized in Table \ref{table:datasets} (appendix). 

For a fair comparison, for all algorithms considered, \textbf{we define
an iteration as a run of $m$ edge updates}, and each update can be
implemented in a constant time. 

We plot the best density obtained by each algorithm over the iterations
in Figure \ref{fig:iteration-best-density}. In Figure \ref{fig:time-best-density},
we plot the best density over wall clock time. Finally, Figure \ref{fig:iteration-function-value}
shows the function value ($L_{2}$-norm of the load vector) $\left(\sum_{u\in V}b_{u}^{2}\right)^{1/2}$
over the iterations. Due to space limit, we only show plots for two
datasets: com-Amazon and orkut. We also exclude Frank-Wolfe in the
plots due to its significantly worse performance in all instances.
We defer the remaining plots and plots that include Frank-Wolfe to
the appendix. 

\textbf{Discussion.$\quad$} We can observe that Algorithms \ref{alg:acdm}
and \ref{alg:rcdm} are practical and can run on relatively fast large
instances (for example, orkut has more than $3$ million vertices
and $100$ million edges). Figure \ref{fig:iteration-function-value}
shows that both Algorithms \ref{alg:acdm} and \ref{alg:rcdm} outperform
the others at minimizing the function value. Especially in comparison
with FISTA, both Algorithms \ref{alg:acdm} and \ref{alg:rcdm} are
significantly better.

\citet{boob2020flowless} observed that in most instances, the Greedy
peeling algorithm by \citet{charikar2000greedy} already finds a near-optimal
densest subgraph. Greedy++ inherits this feature of Greedy and generally
has a very good performance across instances. Algorithms \ref{alg:acdm}
and \ref{alg:rcdm} with initialization by the Greedy algorithm have
competitive performances with Greedy++ and FISTA both in terms of
the number of iterations and time.

\section{Conclusion}

In this paper, we present several algorithms for the DSG problems.
We show new algorithms via multiplicative weights update and area
convexity with improved running times. We also give the first practical
algorithm with a linear convergence rate via random coordinate descent.
Obtaining a practical implementation of our multiplicative weights
update algorithm in the streaming and distributed settings, and using
our results to improve algorithms for DSG problems in other settings
such as differential privacy are among potential future works. 

\section*{Acknowledgements}

The authors were supported in part by NSF CAREER grant CCF-1750333,
NSF grant III-1908510, and an Alfred P. Sloan Research Fellowship.
The authors thank Cheng-Hao Fu for helpful discussions in the preliminary
stages of this work.

\bibliographystyle{plainnat}
\bibliography{ref}

\appendix

\section{Additional Proofs from Section \ref{sec:Multiplicative-weights-update}\label{sec:appdx-MWU}}

\subsection{Multiplicative Weights Update analysis tool: $\protect\smax$ function
\label{subsec:MWU-tool}}

We analyze our MWU algorithm via the $\smax$ function, defined as
follows. For $x\in\R^{n}$ and $\eta\in\R_{+}$, we define
\[
\smax_{\eta}(x)=\frac{1}{\eta}\ln\left(\sum_{i=1}^{n}\exp\left(\eta x_{i}\right)\right).
\]
$\smax_{\eta}(x)$ can be seen as a smooth approximation of $\max(x)\overset{\mathrm{def}}{=}\max_{i}x_{i}\leq\max_{i}\left|x_{i}\right|=\left\Vert x\right\Vert _{\infty}$,
in the following sense
\begin{align*}
\left\Vert x\right\Vert _{\infty} & \le\smax_{\eta}(x)\leq\frac{\ln n}{\eta}+\left\Vert x\right\Vert _{\infty},
\end{align*}
and $\smax_{\eta}$ is $\eta$-smooth with respect to $\left\Vert \cdot\right\Vert _{\infty}$:
$\forall x,u,$
\[
\smax_{\eta}(x+u)\leq\smax_{\eta}(x)+\left\langle \nabla\smax_{\eta}(x),u\right\rangle +\text{\ensuremath{\frac{\eta}{2}}}\left\Vert u\right\Vert _{\infty}^{2}.
\]
The gradient of $\smax_{\eta}$ is a probability distribution in $\Delta_{n}=\left\{ p\in\R_{\geq0}^{n}\colon p_{1}+\dots+p_{n}=1\right\} $:
\[
\left(\nabla\smax_{\eta}(x)\right)_{i}=\text{\ensuremath{\frac{\exp(\eta x_{i})}{\sum_{j}\exp(\eta x_{j})}}}.
\]

\subsection{Proof of Lemma \ref{lem:mwu-average-solution} }

\begin{proof}
Note that $z^{*}$ can be decreased so that $z_{eu}^{*}+z_{ev}^{*}=1,\ \forall e=uv\in E$,
without increasing the objective. Hence, we can have $\sum_{e\in E}p_{e}^{(t)}\left(z_{eu}^{*}+z_{ev}^{*}\right)=\sum_{e\in E}p_{e}^{(t)}=1$.
This means $z^{*}$ satisfies the constraint of LP (\ref{eq:dual-wr-iteration}),
thus for all $t$, since $z^{(t)}$ is an optimal solution to (\ref{eq:dual-wr-iteration})
with $p^{(t)}$, we have
\begin{align*}
\max_{v\in V}\sum_{e\in E,u\in e}z_{eu}^{(t)} & \le\max_{v\in V}\sum_{e\in E,u\in e}z_{eu}^{*},
\end{align*}
which implies 
\begin{align*}
\max_{v\in V}\sum_{e\in E,u\in e}\overline{z}_{eu} & =\max_{v\in V}\frac{1}{T}\sum_{t=1}^{T}\sum_{e\in E,u\in e}z_{eu}^{(t)}\\
 & \le\frac{1}{T}\sum_{t=1}^{T}\max_{v\in V}\sum_{e\in E,u\in e}z_{eu}^{(t)}\\
 & \le\max_{v\in V}\sum_{e\in E,u\in e}z_{eu}^{*}.
\end{align*}
Moreover, since $\smax_{\eta}$ is $\eta$-smooth wrt $\left\Vert \cdot\right\Vert _{\infty}$,
we have
\begin{align*}
\smax_{\eta}(G^{(t)})-\smax_{\eta}(G^{(t-1)}) & \leq\left\langle \nabla\smax_{\eta}(G^{(t-1)}),g^{(t)}\right\rangle +\frac{\eta}{2}\left\Vert g^{(t)}\right\Vert _{\infty}^{2}\\
 & =\sum_{e\in E}p_{e}^{(t)}\left(1-\left(z_{eu}^{(t)}+z_{ev}^{(t)}\right)\right)+\frac{\eta}{2}\left\Vert g^{(t)}\right\Vert _{\infty}^{2}\\
 & =\frac{\eta}{2}\left\Vert g^{(t)}\right\Vert _{\infty}^{2},
\end{align*}
where we use $g^{(t)}=1-\left(z_{eu}^{(t)}+z_{ev}^{(t)}\right)$ and
$\nabla\smax_{\eta}(G^{(t-1)})=p^{(t)}$. Note that $g_{e}^{(t)}=1-\left(z_{eu}^{(t)}+z_{ev}^{(t)}\right)\in\left[-1,1\right].$Thus
\begin{align*}
\smax_{\eta}(G^{(T)}) & \leq\underbrace{\smax_{\eta}(G^{(0)})}_{=\frac{\ln m}{\eta}}+\sum_{t=1}^{T}\frac{\eta}{2}\underbrace{\left\Vert g^{(t)}\right\Vert _{\infty}^{2}}_{\le1}\leq\frac{\ln m}{\eta}+\frac{\eta}{2}T.
\end{align*}
Thus we have
\[
\max\left(G^{(T)}\right)\leq\smax_{\eta}(G^{(T)})\leq\frac{\ln m}{\eta}+\frac{\eta}{2}T
\]
and hence
\[
\max\left(\left(1-\left(\overline{z}_{eu}+\overline{z}_{ev}\right)\right)_{e\in E}\right)\leq\frac{1}{T}\frac{\ln m}{\eta}+\frac{\eta}{2}.
\]
By the choice $\eta=\epsilon$ and $T=\frac{2\ln m}{\epsilon\eta}=\frac{2\ln m}{\epsilon^{2}}$,
we obtain $\frac{1}{T}\frac{\ln m}{\eta}+\frac{\eta}{2}\leq\epsilon$,
i.e.,
\[
1-\left(\overline{z}_{eu}+\overline{z}_{ev}\right)\leq\epsilon\quad\forall e=uv.
\]
\end{proof}

\subsection{Proof of Corollary \ref{cor:existence-good-D}}

\begin{proof}
Assume the contradiction: for all $t\in[T]$ we have $D^{(t)}<\left(1-\epsilon\right)D^{*},$
which means $\max_{v\in V}\sum_{e\in E,u\in e}z_{eu}^{(t)}<\left(1-\epsilon\right)D^{*}.$
Then 
\begin{align*}
\max_{v\in V}\sum_{e\in E,u\in e}\overline{z}_{eu} & =\max_{v\in V}\frac{1}{T}\sum_{t=1}^{T}\sum_{e\in E,u\in e}z_{eu}^{(t)}\\
 & \le\frac{1}{T}\sum_{t=1}^{T}\max_{v\in V}\sum_{e\in E,u\in e}z_{eu}^{(t)}\\
 & <\left(1-\epsilon\right)D^{*}.
\end{align*}
On the other hand, we have for all $e=uv$, $\overline{z}_{eu}+\overline{z}_{ev}\ge1-\epsilon$.
We let $\widetilde{z}_{eu}=\frac{\overline{z}_{eu}}{\overline{z}_{eu}+\overline{z}_{ev}}$.
In this way we have $\widetilde{z}_{eu}+\widetilde{z}_{ev}=1$ and
$\widetilde{z}_{eu}\le\frac{\overline{z}_{eu}}{1-\epsilon}$. Therefore
$\widetilde{z}$ satisfies the constraint of (\ref{eq:dual-wr}).
Furthermore 
\begin{align*}
\max_{v\in V}\sum_{e\in E,u\in e}\widetilde{z}_{eu} & \le\frac{1}{1-\epsilon}\max_{v\in V}\sum_{e\in E,u\in e}\overline{z}_{eu}<D^{*}.
\end{align*}
which means $\widetilde{z}$ has a better objective than $z^{*}$,
contradiction.
\end{proof}

\subsection{Proof of Lemma \ref{lem:assignment}}

\begin{proof}
First, $z^{*}$ satisfies $\sum_{e\in E,u\in e}z_{eu}^{*}\le D^{*}$
for all $u$. Assume that $\widetilde{z}$ is an optimal solution
to (\ref{eq:dual-wr-iteration}). For each $u$ we show that 
\begin{align}
\sum_{e\in E,u\in e}p_{e}z_{eu}^{*} & \ge\sum_{e\in E,u\in e}p_{e}\widetilde{z}_{eu}.\label{eq:3}
\end{align}
Indeed, if $\deg u\le\lfloor D^{*}\rfloor$ we have $z_{eu}^{*}=1$
for all $e\ni u$. Hence (\ref{eq:3}) holds. Otherwise we have 
\begin{align*}
\sum_{e\in E,u\in e}z_{eu}^{*} & =D^{*}\ge\sum_{e\in E,u\in e}\widetilde{z}_{eu}.
\end{align*}
Furthermore $z^{*}$ satisfies $z_{eu}^{*}\ge z_{e'u}^{*}$ if $p_{e}\ge p_{e'}$,
hence $\sum_{e\in E,u\in e}p_{e}z_{eu}^{*}$ maximizes $\sum_{e\in E,u\in e}p_{e}z_{eu}$
subject to $\sum_{e\in E,u\in e}z_{eu}\le D^{*}$. Thus (\ref{eq:3})
holds. Therefore 
\begin{align*}
\sum_{u\in V}\sum_{e\in E,u\in e}p_{e}z_{eu}^{*} & \ge\sum_{u\in V}\sum_{e\in E,u\in e}p_{e}\widetilde{z}_{eu}=1.
\end{align*}
If $\sum_{u\in V}\sum_{e\in E,u\in e}p_{e}z_{eu}^{*}>1$, we can decrease
the value of $z$ for all the vertices where $\sum_{e\in E,u\in e}z_{eu}^{*}=D^{*}$
thus obtains a solution with strictly better objective than $D^{*}$,
which is a contradiction. Therefore $z^{*}$ is an optimal solution.
\end{proof}

\subsection{Proof of Lemma \ref{lem:mwu-complementary-slackness}}

\begin{proof}
We verify by complementary slackness.

1) We have$\sum_{v}x_{v}=1$.

2) $\sum_{e\ni u}z_{eu}^{(\tau)}<D^{(\tau)}\Leftrightarrow u\notin X\Leftrightarrow x_{u}=0$.

3) For all $u\notin X$ we have $z_{eu}^{(\tau)}=1$ for all $e\ni u$
and $\alpha_{eu}=p_{e}^{(\tau)}W$. For $u\in X$ such that $z_{eu}^{(\tau)}=0$,
we have $p_{e}^{(\tau)}\le p_{e(u)}^{(\tau)}$, so $\alpha_{eu}=0$.
For $u\in X$, $z_{eu}^{(\tau)}=z_{e(u)u}^{(\tau)}$, we have also
$\alpha_{eu}^{(\tau)}=0$. For $u\in X$ such that $z_{eu}^{(\tau)}>0$,
we also guarantee $p_{e}^{(\tau)}W+\alpha_{eu}^{(\tau)}=x_{u}$. $p_{e}^{(\tau)}W<x_{u}+\alpha_{eu}$
happens only when $p_{e}^{(\tau)}<p_{e(u)}^{(\tau)}$ which gives
$z_{eu}=0$.
\end{proof}

\subsection{Proof of Lemma \ref{lem:primal-optimal}}

\begin{proof}
By strong duality we have
\begin{align*}
W-\sum_{e=uv}\left(\alpha_{eu}+\alpha_{ev}\right) & =D^{(\tau)}.
\end{align*}
We know that $D^{(\tau)}\ge\left(1-\epsilon\right)D^{*}=\left(1-\epsilon\right)\opt$.
Hence $\sum_{e}\left(p_{e}^{(\tau)}W-\left(\alpha_{eu}+\alpha_{ev}\right)\right)=W-\sum_{e=uv}\left(\alpha_{eu}+\alpha_{ev}\right)\ge\left(1-\epsilon\right)\opt.$
On the other hand, since 
\begin{align*}
p_{e}^{(\tau)}W-\left(\alpha_{eu}+\alpha_{ev}\right) & \le p_{e}^{(\tau)}W-\alpha_{eu}\le x_{u},\\
p_{e}^{(\tau)}W-\left(\alpha_{eu}+\alpha_{ev}\right) & \le p_{e}^{(\tau)}W-\alpha_{ev}\le x_{v}.
\end{align*}
We have $\sum_{e=uv}\min\{x_{u},x_{v}\}\ge\sum_{e=uv}(p_{e}^{(\tau)}W-(\alpha_{eu}+\alpha_{ev}))\ge(1-\epsilon)\opt,$
as needed.
\end{proof}

\section{Additional Proofs from Section \ref{sec:area-convex}}

\subsection{Area convexity functions review\label{subsec:appdx-Area-convexity}}

We first review the notion of area convexity introduced by \citet{sherman2017area}.
\begin{defn}
A function $\phi$ is area convex with respect to an anti-symmetric
matrix $A$ on a convex set $K$ if for every $x,y,z\in K$,
\begin{align*}
\phi\left(\frac{x+y+z}{3}\right) & \le\frac{1}{3}\left(\phi\left(x\right)+\phi\left(y\right)+\phi\left(z\right)\right)-\frac{1}{3\sqrt{3}}\left(x-y\right)^{T}A\left(y-z\right).
\end{align*}
 To show that a function is area convex, \citet{boob2019faster} employ
operator $\succeq_{i}$. For a symmetric matrix $A$ and an anti-symmetric
matrix $B$, we say $A\succeq_{i}B$ iff \textbf{$\left[\begin{array}{cc}
A & -B^{T}\\
B & A
\end{array}\right]$} is PSD. The following two lemmas are from \citet{boob2019faster}.
\end{defn}

\begin{lem}
\label{lem:av-succi}(Lemma 4.5 in \citet{boob2019faster}) Let $A$
be a $\R^{2\times2}$ symmetric matrix. $A\succeq_{i}\left[\begin{array}{cc}
0 & -1\\
1 & 0
\end{array}\right]$ iff $A\succeq0$ and $\det A\ge1$.
\end{lem}

\begin{lem}
\label{lem:av-hessian} (Lemma 4.6 in \citet{boob2019faster}) Let
$\phi$ be twice differentiable on the interior of convex set $K$,
i.e $\mathrm{int}(K)$. If $\nabla^{2}\phi(x)\succeq_{i}A$ for all
$x\in\mathrm{int}(K)$ then $\phi$ is area convex with respect to
$\frac{1}{3}A$ on $\mathrm{int}(K)$. If moreover, $\phi$ is continuous
on $\mathrm{cl}(K)$ then $\phi$ is area convex with respect to $\frac{1}{3}A$
on $\mathrm{cl}(K)$.
\end{lem}

\subsection{Reduction to the saddle point problem}
\begin{lem}
\label{lem:ac-reduction}(Lemma 4.3 in \citet{boob2019faster}) Suppose
$z\in C,y\in\Delta_{m}$ satisfy
\begin{align*}
\max_{\overline{z}\in C,\overline{y}\in\Delta_{m}}\sum_{e}y_{e}\left(\overline{z}_{eu}+\overline{z}_{ev}\right)-\overline{y}_{e}\left(z_{eu}+z_{ev}\right) & \le\epsilon,
\end{align*}
then either of the following happens:

1. $z$ is an $\epsilon$-approximate solution to the feasibility
problem,

2. $y$ satisfies for all $\overline{z}\in C$, $\sum_{e}y_{e}\left(\overline{z}_{eu}+\overline{z}_{ev}\right)<1$.
\end{lem}

\begin{proof}
Suppose that $z$ is not an $\epsilon$-approximate solution to the
problem. This means there exists $e=uv$ such that $z_{eu}+z_{ev}<1-\epsilon$,
which implies 
\begin{align*}
\min_{\overline{y}\in\Delta_{m}}\overline{y}_{e}\left(z_{eu}+z_{ev}\right) & <1-\epsilon
\end{align*}
Since 
\begin{align*}
 & \max_{\overline{z}\in C,\overline{y}\in\Delta_{m}}\sum_{e}y_{e}\left(\overline{z}_{eu}+\overline{z}_{ev}\right)-\overline{y}_{e}\left(z_{eu}+z_{ev}\right)\\
= & \max_{\overline{z}\in C}\sum_{e}y_{e}\left(\overline{z}_{eu}+\overline{z}_{ev}\right)-\min_{\overline{y}\in\Delta_{m}}\overline{y}_{e}\left(z_{eu}+z_{ev}\right)\le\epsilon,
\end{align*}
we can conclude that
\begin{align*}
\max_{\overline{z}\in C}\sum_{e}y_{e}\left(\overline{z}_{eu}+\overline{z}_{ev}\right) & <1.
\end{align*}
\end{proof}

\subsection{Properties of the regularizer }

We recall the choice of the regularizer function
\begin{align*}
\phi(z,y) & =6\sqrt{3}\left(\sum_{e\in E}y_{e}\left(z_{eu}^{2}+z_{ev}^{2}\right)+6y_{e}\log y_{e}-2\right).
\end{align*}
Our goal is to show that $\phi$ is area convex with respect to $A$
and has a small range. 
\begin{lem}
\label{lem:av-proof-of-convex}$\frac{1}{6\sqrt{3}}\phi$ is area
convex with respect to $\frac{1}{3}A$. Furthermore $-6\sqrt{3}\left(6\log m+2\right)\le\phi(z,y)\le0$.
\end{lem}

\begin{proof}
By Lemma \ref{lem:av-hessian}, it suffices to show that
\begin{align*}
\nabla^{2}\phi(z,y) & \succeq_{i}A.
\end{align*}
Let $\vec{f}_{e}$ denote the vector with all $0$'s and one $1$
at the index of $y_{e}$, $\vec{f}_{eu}$ denote the vector with all
$0$'s and one $1$ at the index of $z_{ev}$ and $\vec{f}_{u}$ for
$u$. Consider two variables $y_{e}$ and $z_{eu}$, we have

\begin{align*}
\nabla^{2}y_{e}\left(\left(z_{eu}^{2}+3\log y_{e}\right)\right) & =\left[\begin{array}{cc}
\frac{3}{y_{e}} & 2z_{eu}\\
2z_{eu} & 2y_{e}
\end{array}\right]\succeq_{i}\left[\begin{array}{cc}
0 & -1\\
1 & 0
\end{array}\right]
\end{align*}
where the last inequality comes from Lemma \ref{lem:av-succi} and
that
\begin{align*}
\det\left[\begin{array}{cc}
\frac{3}{y_{e}} & 2z_{eu}\\
2z_{eu} & 2y_{e}
\end{array}\right] & =6-4z_{eu}^{2}>1,
\end{align*}
which holds because $0\le z_{eu}\le1$. By Lemma 4.10 in \citet{boob2019faster}
\begin{align*}
 & \nabla^{2}y_{e}\left(\left(z_{eu}^{2}+3\log y_{e}\right)+\left(z_{ev}^{2}+3\log y_{e}\right)\right)\\
\succeq_{i} & \left(\vec{f}_{e}\vec{f}_{eu}^{T}-\vec{f}_{eu}\vec{f}_{e}^{T}\right)+\left(\vec{f}_{e}\vec{f}_{ev}^{T}-\vec{f}_{ev}\vec{f}_{e}^{T}\right).
\end{align*}
Sum up the RHS we get exactly $A.$

For the lower bound
\begin{align*}
 & \frac{1}{6\sqrt{3}}\phi(z,y)=\sum_{e}y_{e}\left(z_{eu}^{2}+z_{ev}^{2}\right)+6\sum_{e}y_{e}\log y_{e}-2\\
\ge & 6\underbrace{\sum_{e}y_{e}\log y_{e}}_{\mbox{convexity}}-2\ge6m\frac{1}{m}\log\frac{1}{m}-2=-6\log m-2.
\end{align*}
For the upper bound
\begin{align*}
\frac{1}{6\sqrt{3}}\phi(z,y) & \le2\sum_{e}y_{e}-2=0.
\end{align*}
\end{proof}

\subsection{Proof of Lemma \ref{lem:av-convergence}}

\begin{proof}
The proof of this Lemma directly follows from Theorem 1.3 in \citet{sherman2017area}
and that 
\begin{align*}
\phi^{*}(w^{(0)}) & =\sup_{w\in C\times\Delta_{m}}\left\langle w,w^{(0)}\right\rangle -\phi(w)\\
 & =\sup_{w\in C\times\Delta_{m}}1-\phi(w)\\
 & \le1+6\sqrt{3}\left(6\log m+2\right).
\end{align*}
\end{proof}

\subsection{Proof of Lemma \ref{lem:alternating-min-lemma}}

\begin{proof}
The proof of Lemma \ref{lem:alternating-min-lemma} follows from the
general framework for analyzing alternating minimization by \citet{beck2015convergence}.
The proof detail below follows from \citet{jambulapati2019direct}.

For simplicity, let us recall the definition of $H$ in Algorithm
\ref{alg:oracle}, after scaling by $\frac{1}{6\sqrt{3}}$. Given
$x=(s,r)$ is the input, we have 
\begin{align*}
H(z,y)\coloneqq & \sum_{e=uv\in E}y_{e}\left(z_{eu}^{2}+z_{ev}^{2}\right)+6\sum_{e}y_{e}\log y_{e}\\
 & -\frac{1}{6\sqrt{3}}\Big(\sum_{e=uv\in E}(z_{eu}s_{eu}+z_{ev}s_{ev}+y_{e}r_{e})\Big).
\end{align*}

Let $\nabla_{zz}^{2}$ be the Hessian with all but the $zz$ block
zeroed out. We use $\nabla_{y}$ and $\nabla_{z}$ to denote the gradient
with only the $y$ and $z$ components kept.

Let $Y^{(t+1)}=\left\{ y\in\Delta_{m}:y\ge\frac{1}{2}y^{(t+1)}\right\} $.
We will first show that for all $z,\overline{z}\in C$ and $\overline{y}\in Y^{(t+1)}$
\begin{align}
\nabla^{2}H(\overline{z},\overline{y}) & \succeq\frac{1}{6}\nabla_{zz}^{2}H(z,y^{(t+1)})\label{eq:hessian-bound}
\end{align}

Since we do not have any cross term between $e$ and $e'$ for any
$e\neq e'$ we can consider edge separately. For the same reason,
we can also separate $z_{eu}$ and $z_{ev}$ for each edge $e$. The
non-zero term after taking the Hessian for edge $e$ and vertex $u$
is
\begin{align*}
\nabla^{2}\overline{y}_{e}\left(\left(\overline{z}_{eu}^{2}+3\log\overline{y}_{e}\right)\right) & =\left[\begin{array}{cc}
\frac{3}{\overline{y}_{e}} & 2\overline{z}_{eu}\\
2\overline{z}_{eu} & 2\overline{y}_{e}
\end{array}\right];\\
\nabla_{zz}^{2}y_{e}^{(t+1)}\left(\left(z_{eu}^{2}+3\log y_{e}^{(t+1)}\right)\right) & =\left[\begin{array}{cc}
0 & 0\\
0 & 2y_{e}^{(t+1)}
\end{array}\right].
\end{align*}
For all $a,b\in\R$
\begin{align*}
\left[\begin{array}{cc}
a & b\end{array}\right]\nabla^{2}\overline{y}_{e}\left(\left(\overline{z}_{eu}^{2}+3\log\overline{y}_{e}\right)\right)\left[\begin{array}{c}
a\\
b
\end{array}\right]= & \frac{3}{\overline{y}_{e}}a^{2}+2\overline{y}_{e}b^{2}+4\overline{z}_{eu}ab;\\
\left[\begin{array}{cc}
a & b\end{array}\right]\nabla_{zz}^{2}y_{e}^{(t+1)}\left(\left(z_{eu}^{2}+3\log y_{e}^{(t+1)}\right)\right)\left[\begin{array}{c}
a\\
b
\end{array}\right]= & 2y_{e}^{(t+1)}b^{2}.
\end{align*}
Since $z_{eu},\overline{z}_{eu}\le1$
\begin{align*}
\frac{3}{\overline{y}_{e}}a^{2}+2\overline{y}_{e}b^{2}+4\overline{z}_{eu}ab & \ge\frac{3}{\overline{y}_{e}}a^{2}+2\overline{y}_{e}b^{2}-4\left|ab\right|\\
 & \ge\frac{2}{3}\overline{y}_{e}b^{2}\\
 & \ge\frac{1}{3}y_{e}^{(t+1)}b^{2}\mbox{ for all }y\ge\frac{1}{2}y^{(t+1)}\\
 & =\frac{1}{6}\times2y_{e}^{(t+1)}b^{2}.
\end{align*}
Hence for all $y\ge\frac{1}{2}y^{(t+1)}$ 
\begin{align*}
\nabla^{2}\overline{y}_{e}\left(\overline{z}_{eu}^{2}+3\log\overline{y}_{e}\right) & \succeq\nabla_{zz}^{2}y_{e}^{(t+1)}\left(z_{eu}^{2}+3\log y_{e}^{(t+1)}\right)
\end{align*}
which gives us (\ref{eq:hessian-bound}).

Now we show that or all $y^{*}\in Y^{(t+1)}$ and $z^{*}\in C$
\begin{align*}
H(z^{(t)},y^{(t+1)})-H(z^{(t+1)},y^{(t+1)}) & \ge\frac{1}{6}\left(H(z^{(t)},y^{(t+1)})-H(z^{*},y^{*})\right).
\end{align*}
Let $\tilde{z}=\frac{5}{6}z^{(t)}+\frac{1}{6}z^{*}$. By the definition
of $z^{(t+1)}$ we have 
\begin{align*}
H(z^{(t+1)},y^{(t+1)}) & \le H(\tilde{z},y^{(t+1)}.)
\end{align*}
By the optimality of $y^{(t+1)}$ and the convexity of $H$
\begin{align*}
\left\langle \nabla_{y}H(z^{(t)},y^{(t+1)}),y^{(t+1)}-y^{*}\right\rangle  & \le0
\end{align*}
which gives us
\begin{align*}
\left\langle \nabla_{z}H(z^{(t)},y^{(t+1)}),z^{(t)}-\tilde{z}\right\rangle = & \frac{1}{6}\left\langle \nabla_{z}H(z^{(t)},y^{(t+1)}),z^{(t)}-z^{*}\right\rangle \\
\ge & \frac{1}{6}\left\langle \nabla_{z}H(z^{(t)},y^{(t+1)}),z^{(t)}-z^{*}\right\rangle \\
 & +\frac{1}{6}\left\langle \nabla_{y}H(z^{(t)},y^{(t+1)}),y^{(t+1)}-y^{*}\right\rangle \\
= & \frac{1}{6}\left\langle \nabla H(z^{(t)},y^{(t+1)}),w^{(t+\frac{1}{2})}-w^{*}\right\rangle 
\end{align*}
where $w^{(t+\frac{1}{2})}=(z^{(t)},y^{(t+1)})$, $w^{*}=(z^{*},y^{*})$.
Also define $z_{\alpha}=\left(1-\alpha\right)z^{(t)}+\alpha z^{*}$,
$\tilde{z}_{\alpha}=\left(1-\alpha\right)z^{(t)}+\alpha\tilde{z}$,
$y_{\alpha}=\left(1-\alpha\right)y^{(t+1)}+\alpha y^{*}$. With a
slight abuse of notion, we also use $\nabla_{zz}^{2}$ to also mean
the Hessian with respect to the variable $z$. Using Taylor expansion
\begin{align*}
H(z^{(t)},y^{(t+1)})-H(\tilde{z},y^{(t+1)})= & \left\langle \nabla_{z}H(z^{(t)},y^{(t+1)}),z^{(t)}-\tilde{z}\right\rangle \\
 & -\int_{0}^{1}\int_{0}^{\beta}\left(\tilde{z}-z^{(t)}\right)^{T}\nabla_{zz}^{2}H(\tilde{z}_{\alpha},y^{(t+1)})\left(\tilde{z}-z^{(t)}\right)d\alpha d\beta\\
\ge & \frac{1}{6}\left\langle \nabla H(z^{(t)},y^{(t+1)}),w^{(t+\frac{1}{2})}-w^{*}\right\rangle \\
 & -\frac{1}{36}\int_{0}^{1}\int_{0}^{\beta}\left(z^{*}-z^{(t)}\right)^{T}\nabla_{zz}^{2}H(\tilde{z}_{\alpha},y^{(t+1)})\left(z^{*}-z^{(t)}\right)d\alpha d\beta\\
\ge & \frac{1}{6}\left\langle \nabla H(z^{(t)},y^{(t+1)}),w^{(t+\frac{1}{2})}-w^{*}\right\rangle \\
 & -\frac{1}{6}\int_{0}^{1}\int_{0}^{\beta}\left(w^{*}-w^{(t+\frac{1}{2})}\right)^{T}\nabla^{2}H(z_{\alpha},y_{\alpha})\left(w^{*}-w^{(t+\frac{1}{2})}\right)d\alpha d\beta\\
 & =\frac{1}{6}\left(H(z^{(t)},y^{(t+1)})-H(z^{*},y^{*})\right).
\end{align*}
Hence 
\begin{align*}
H(z^{(t)},y^{(t+1)})-H(z^{(t+1)},y^{(t+1)})\ge & H(z^{(t)},y^{(t+1)})-H(\tilde{z},y^{(t+1)})\\
\ge & \frac{1}{6}\left(H(z^{(t)},y^{(t+1)})-H(z^{*},y^{*})\right).
\end{align*}
Take $z^{*}=\frac{1}{2}\left(z^{(t)}+z_{\opt}\right)$, $y^{*}=\frac{1}{2}\left(y^{(t+1)}+y_{\opt}\right)$
\begin{align*}
H(z^{(t)},y^{(t+1)})-H(z^{(t+1)},y^{(t+2)})\ge & H(z^{(t)},y^{(t+1)})-H(z^{(t+1)},y^{(t+1)})\\
\ge & \frac{1}{6}\left(H(z^{(t)},y^{(t+1)})-H(z^{*},y^{*})\right)\\
\ge & \frac{1}{6}\left(H(z^{(t)},y^{(t+1)})-\left(\frac{1}{2}H(z^{(t)},y^{(t+1)})+\frac{1}{2}H\left(z_{\opt},y_{\opt}\right)\right)\right)\\
 & \qquad\qquad(\mbox{by convexity of }H)\\
= & \frac{1}{12}\left(H(z^{(t)},y^{(t+1)})-H\left(z_{\opt},y_{\opt}\right)\right)
\end{align*}
which means
\begin{align*}
H(z^{(t+1)},y^{(t+2)})-H\left(z_{\opt},y_{\opt}\right) & \le\frac{11}{12}\left(H(z^{(t)},y^{(t+1)})-H\left(z_{\opt},y_{\opt}\right)\right).
\end{align*}
Therefore 
\begin{align*}
H(z^{(T+1)},y^{(T+1)})-H\left(z_{\opt},y_{\opt}\right) & \le\left(\frac{11}{12}\right)^{T}\left(H(z^{(0)},y^{(1)})-H\left(z_{\opt},y_{\opt}\right)\right).
\end{align*}
This gives us the convergence rate.
\end{proof}

\subsection{Proof of Lemma \ref{lem:av-iteration-time}\label{subsec:av-iter-time}}

\begin{proof}
For the first minimization, we have $y^{(t+1)}=\arg\min_{y\in\Delta_{m}}H(z^{(t)},y)=\arg\max_{y\in\Delta_{m}}\left\langle L^{(t)},y\right\rangle -\sum_{e}y_{e}\log y_{e}$,
where $L_{e}^{(t)}=-\frac{1}{6}\left(\left(z_{eu}^{(t)2}+z_{ev}^{(t)2}\right)-\frac{1}{6\sqrt{3}}r_{e}\right)$.
The solution is simply $\nabla\smax(L^{(t)})$ (definition in Section
\ref{subsec:MWU-tool}), which can be computed in $O(m)$.

The second minimization $z^{(t+1)}=\arg\min_{z\in C}H(z,y^{(t+1)})$.
Here we build on the insights from the oracle implementation for MWU
and reduce the problem to computing for each $u$ separately
\begin{align*}
\min_{z\in[0,1]^{\deg(u)}} & \sum_{e\ni u}z_{eu}^{2}y_{e}^{(t+1)}-\frac{1}{6\sqrt{3}}z_{eu}s_{eu}\mbox{ st. }\sum_{e\ni u}z_{eu}\le D
\end{align*}
Let $\tilde{s}=\frac{1}{6\sqrt{3}}s$ and take the Lagrangian, we
have
\begin{align*}
\min_{z\in[0,1]^{\deg(u)}}\max_{\lambda\ge0} & \sum_{e\ni u}\left(z_{eu}^{2}y_{e}^{(t+1)}-z_{eu}\tilde{s}_{eu}+\lambda z_{eu}\right)-\lambda D\\
\Leftrightarrow\max_{\lambda\ge0}\min_{z\in[0,1]^{\deg(u)}} & \sum_{e\ni u}\left(z_{eu}^{2}y_{e}^{(t+1)}-z_{eu}\tilde{s}_{eu}+\lambda z_{eu}\right)-\lambda D
\end{align*}
For $\lambda\ge0$, we obtain for each $e\ni u$
\begin{align*}
z_{eu} & =\max\left\{ 0,\min\left\{ 1,\frac{\tilde{s}_{eu}-\lambda}{2y_{e}^{(t+1)}}\right\} \right\} 
\end{align*}
Now we need to solve for $\lambda$
\begin{align*}
 & \max_{\lambda\ge0}\sum_{e\ni u}\left(z_{eu}^{2}y_{e}^{(t+1)}-z_{eu}\tilde{s}_{eu}+\lambda z_{eu}\right)-\lambda D\\
= & \max_{\lambda\ge0}-\lambda D+\sum_{z:0\le\tilde{s}_{eu}-\lambda\le2y_{e}^{(t+1)}}-\frac{\left(\tilde{s}_{eu}-\lambda\right)^{2}}{4y_{e}^{(t+1)}}\\
 & \quad+\sum_{z:\tilde{s}_{eu}-\lambda>2y_{e}^{(t+1)}}\left(y_{e}^{(t+1)}-\tilde{s}_{eu}+\lambda\right).
\end{align*}
Again, we take the inspiration from Algorithm \ref{alg:iteration-solver}
and see that we can also perform a search for $\lambda$. Here each
$e\ni u$ belongs to one of the three category, $\tilde{s}_{eu}-\lambda<0$
or $0\le\tilde{s}_{eu}-\lambda\le2y_{e}^{(t+1)}$ or $\tilde{s}_{eu}-\lambda>2y_{e}^{(t+1)}$.
To solve the above problem, we must determine which category each
$e$ belongs to. To do this, we can sort 2$\deg u$ numbers $\left\{ \max\left\{ 2y_{e}^{(t+1)}-\tilde{s}_{eu},0\right\} ,\tilde{s}_{eu}\right\} _{e\ni u}$
and find the optimal value of $\lambda$ on each interval. When testing
$\lambda$ increasingly, the category of each $z_{eu}$ only changes
at most twice. For this we can use a data structure (eg. Fibonacci
heap) to determine which $z_{eu}$ changes category when $\lambda$
jumps to the next interval. This means the total time to find $z_{eu}$
for all $e\ni u$ is at most $O(\deg u\log\deg u)$. Summing the total
over all vertex $u$, we have solving the second minimization problem
each iteration takes $O(m\log\Delta)$ time.
\end{proof}

\subsection{Proof of Lemma \ref{lem:av-objective}}

\begin{proof}
Observe that the following LP is infeasible for $z\in[0,1]^{2m}$
\begin{align*}
\sum_{e\ni u}z_{eu} & \le\overline{D},\quad\forall u\in V\\
z_{eu}+z_{ev} & \ge1-\epsilon\quad\forall e=uv\in E.
\end{align*}
Because otherwise, similar to lemma $\ref{cor:existence-good-D}$,
we must have $\overline{D}\ge D^{*}(1-\epsilon)$, while we have $\overline{D}=\tilde{D}(1-2\epsilon)\le D^{*}(1+\epsilon)(1-2\epsilon)<D^{*}(1-\epsilon)$,
contradiction. Thus we have that

\begin{align*}
\max_{z\in C(\overline{D})}\sum_{e}\overline{y}_{e}\left(z_{eu}+z_{ev}\right)\le & \epsilon+\min_{y\in\Delta_{m}}\sum_{e}y_{e}\left(\overline{z}_{eu}+\overline{z}_{ev}\right)<\epsilon+1-\epsilon=1.
\end{align*}
This gives us the claim in the lemma.
\end{proof}

\section{Additional Proofs from Section \ref{sec:DSM}\label{sec:appdx-DSM}}

\subsection{Continuous formulation}

We recall with the quadratic program for finding a dense decomposition
\begin{align}
\min f(z)\coloneqq\sum_{u\in V}b_{u}^{2}\mbox{ st. } & b_{u}=\sum_{e\in E,u\in e}z_{eu},\quad\forall u\in V\label{eq:dual-quadratic-1}\\
 & z_{eu}+z_{ev}\ge1,\quad\forall e=uv\in E\nonumber \\
 & 0\le z_{eu}\le1,\quad\forall e,u\in e\in E.\nonumber 
\end{align}
Now, we show how to reformulate this problem as (\ref{eq:supermod-min}).
Recall that we define for $e\in E$, $F_{e}(S)=1$ if $e\subseteq S$,
$F_{e}(S)=0$ otherwise and the base contrapolymatroid
\begin{align*}
B(F_{e}) & =\left\{ z_{e}\in\R^{n},z_{e}(S)\ge F_{e}(S)\;\forall S\subseteq V,z_{e}(V)=F_{e}(V)=1\right\} 
\end{align*}
Specifically, for $z_{e}\in B(F_{e})$, we have
\begin{align*}
z_{eu}+z_{ev} & =1,\quad\mbox{for }e=uv\\
z_{ew} & =0,\quad\forall w\neq u,v
\end{align*}
In this view, it is immediate to see that we can rewrite the above
problem as
\begin{align}
\min_{z_{e}\in B(F_{e}),\forall e\in E} & \left\Vert \sum_{e\in E}z_{e}\right\Vert _{2}^{2}\label{eq:supermod-min-1}
\end{align}
Following the framework by \citet{ene2015random}, let us write
\begin{align*}
A=\underbrace{\left[I_{n}\dots I_{n}\right]}_{n\mbox{ times}};\quad & {\cal P}=\Pi_{e\in E}B(F_{e})\subseteq\R^{mn}
\end{align*}
The problem can then be written as
\begin{align}
\min_{z\in{\cal P}} & \frac{1}{2}\left\Vert Az\right\Vert _{2}^{2}\label{eq:supermod-prob-1}
\end{align}
The objective function is $2$-smooth with respect to each coordinate.
However, it is not strongly convex. In order to show an algorithm
with linear convergence, our goal is to prove a property similar to
strong convexity.
\begin{defn}[Restricted strong convexity \citep{ene2015random}]
 For $z\in{\cal P}$, let $z^{*}=\arg\min_{p}\left\{ \left\Vert p-z\right\Vert _{2}\colon Ap=b^{*}\right\} $
where $b^{*}$ is the unique optimal solution to (\ref{eq:dual-quadratic}).
We say that $\frac{1}{2}\left\Vert Az\right\Vert _{2}^{2}$ is restricted
$\ell$-strongly convex if for all $y\in{\cal P}$
\begin{align*}
\left\Vert A(z-z^{*})\right\Vert _{2}^{2} & \ge\ell\left\Vert z-z^{*}\right\Vert _{2}^{2}.
\end{align*}
\end{defn}

\begin{lem}
\label{lem:dsm-convexity-1}Let $\ell^{*}=\sup\left\{ \ell:\frac{1}{2}\left\Vert Az\right\Vert _{2}^{2}\mbox{ is restricted }\ell\mbox{-strongly convex}\right\} .$
We have $\ell^{*}\ge\frac{4}{n^{2}}$.
\end{lem}

\begin{proof}
The proof essentially follow from \citet{ene2017decomposable}. For
$b=\sum_{e}z_{e}$ we construct the following directed graph on $G=(V,E)$
and capacities $c$. For $e=uv\in E$, $c(uv)=z_{eu}$, $c(vu)=z_{ev}$.
If an arc has capacity $0$ we just delete the arc from the graph.

We transform $z$ to $y$ that satisfies $Ay=b^{*}$. We initialize
$y=z$. Let $N=\left\{ v:\left(Ay\right)(v)>b^{*}(v)\right\} $ and
$P=\left\{ v:\left(Ay\right)(v)<b^{*}(v)\right\} $. Once we have
$N=P=\emptyset,$ we have $Ay=b^{*}$.
\begin{claim}
If $N\neq\emptyset$ there exists a directed path of positive capacity
between $N$ and $P$.
\end{claim}

\begin{proof}
Let $b=Ay$. Let $S$ be the set of vertices reachable from $N$ on
a directed path of positive capacity. For a contradiction, assume
$S\cap P=\emptyset$. For all $e=uv\subseteq S$ we have $z_{eu}+z_{ev}=1$.
Also there is no out-going edge from $S$ (ie, if there is a edge
$e=uv$ such that $u\in S$ with $v\notin S$, we have $z_{eu}=0$).
By this observation we have
\begin{align*}
b(S) & =\left|S\right|
\end{align*}
On the other hand, since $N\subseteq S$, we have $b(S)=b(N)+b(N\setminus S)>b^{*}(S)+b^{*}(N\setminus S)=b^{*}(S)\ge\left|S\right|$.
So we can conclude that $S\cap P\neq\emptyset$.
\end{proof}
In every step of the algorithm we take the shortest directed path
$p$ of positive capacity from $N$ to $P$ and update $y$. Let $\epsilon$
be the minimum capacity of an arc on $p$. For an arc $(u,v)$, we
update $z_{eu}=z_{eu}-\epsilon$ and $z_{ev}=z_{ev}+\epsilon$. By
doing this, the set of shortest paths of the same length as $p$ strictly
shrinks, until the length of the shortest paths in the graph increases.
For this reason, we know that the algorithm must terminate, which
is when we have $N=P=\emptyset$ and $Ay=b^{*}$.

Every path update changes $\left\Vert y\right\Vert _{\infty}$ at
most $\epsilon$ and $\left\Vert y\right\Vert _{1}$ at most $n\epsilon$.
At the same time $\sum_{v\in N}b(v)-b^{*}(v)$ decreases by $\epsilon$
and $\sum_{v\in P}b^{*}(v)-b(v)$ decreases by $\epsilon$ and $b(v)-b^{*}(v)=0$
for the remaining nodes. Hence $\left\Vert Ay-b^{*}\right\Vert _{1}$
decreases by $2\epsilon$ 
\begin{align*}
\left\Vert z-z^{*}\right\Vert _{\infty} & \le\frac{1}{2}\left\Vert Az-b^{*}\right\Vert _{1}=\frac{1}{2}\left\Vert A(z-z^{*})\right\Vert _{1},\\
\left\Vert z-z^{*}\right\Vert _{1} & \le\frac{n}{2}\left\Vert A(z-z^{*})\right\Vert _{1}.
\end{align*}
Hence we have
\begin{align*}
 & \left\Vert z-z^{*}\right\Vert _{2}^{2}\le\left\Vert z-z^{*}\right\Vert _{\infty}\left\Vert z-z^{*}\right\Vert _{1}\\
\le & \frac{n}{4}\left\Vert A(z-z^{*})\right\Vert _{1}^{2}\le\frac{n^{2}}{4}\left\Vert A(z-z^{*})\right\Vert _{2}^{2}.
\end{align*}
\end{proof}

\subsection{Practical implementation of Accelerated Coordinate Descent}

The implementation of Algorithm \ref{alg:acdm} that we use in our
experiments is shown in Algorithm \ref{alg:acdm-practical}. The main
implementation details are that we select the coordinates via a random
permutation and we restart when the function value increases.

\begin{algorithm}[H]
{\small{}\caption{{\small{}Practical Accelerated Coordinate Descent}}
\label{alg:acdm-practical}}{\small\par}

{\small{}Initialize $z^{(0)}\in{\cal P}$, $b_{u}=\sum_{e\ni u}z_{eu}^{(0)}$,
for all $u$}, $f=\sum_{u\in V}b_{u}^{2}$, $f_{\mathrm{last}}=0$

{\small{}for $k=1\dots K$:}{\small\par}

{\small{}$\quad$for $t=1\dots T$:}{\small\par}

$\quad\quad$if $t=1$ and $f>f_{\mathrm{last}}$:{\small{} $y^{(k,0)}=z^{(k-1)}\in{\cal P}$,
$w^{(k,0)}=0$, $\theta^{(k,1)}=\frac{1}{m}$ // restart when the
function value increases}{\small\par}

$\quad\quad$else: {\small{}$\theta^{(k,t)}=\frac{\sqrt{\theta^{(k,t-1)4}+4\theta^{(k,t-1)2}}-\theta^{(k,t-1)2}}{2}$}{\small\par}

{\small{}$\quad\quad$pick a permutation $R^{(t)}$ of $[m]$}{\small\par}

{\small{}$\quad\quad$for $e\in R^{(t)}:$}{\small\par}

{\small{}$\quad\quad\quad$$x^{(k,t)}=\theta^{(k,t)2}w^{(k,t-1)}+y^{(k,t-1)}$}{\small\par}

{\small{}$\quad\quad\quad$$y^{(k,t)}=\arg\min_{s\in B(F_{e})}\Bigg(\left\langle \nabla_{e}f(x^{(k,t)}),(s_{eu}\ s_{ev})\right\rangle +2m\theta^{(k,t)}\left\Vert (s_{eu}\ s_{ev})-(y_{eu}^{(k,t-1)}\ y_{ev}^{(k,t-1)})\right\Vert _{2}^{2}\Bigg)$}{\small\par}

{\small{}$\quad\quad\quad$$w^{(k,t)}=w^{(k,t-1)}-\frac{1-m\theta^{(k,t-1)}}{\theta^{(k,t-1)2}}\left(y^{(k,t)}-y^{(k,t-1)}\right)$}{\small\par}

{\small{}$\quad\quad$$f_{\mathrm{last}}=f$}{\small\par}

$\quad\quad$update {\small{}$b_{u}=\sum_{e\ni u}\theta^{(k,t)2}w_{e}^{(k,t)}+y_{e}^{(k,t)}$;
$f=\sum_{u\in V}b_{u}^{2}$}{\small\par}

{\small{}$\quad$$z^{(k)}=\theta^{(k,T)2}w^{(k,T)}+y^{(k,T)}$}{\small\par}

{\small{}return $z^{(K)}$}{\small\par}
\end{algorithm}

\subsection{Random Coordinate Descent for solving (\ref{eq:dual-quadratic})}

We also consider random coordinate descent algorithm (the version
of Algorithm \ref{alg:acdm} without acceleration). 

\begin{algorithm}[H]
\caption{Random Coordinate Descent}
\label{alg:rcdm}

Initialize $z^{(0)}\in{\cal P}$

for $t=1\dots T$

$\quad$Option 1: Sample a set $R$ of $m$ edges from $E$ uniformly
at random with replacement

$\quad$Option 2: Pick a random permutation $R$ of $E$

$\quad$for $e\in R$:

$\quad\quad$Update $z^{(k)}=\arg\min_{s\in B(F_{e})}\Bigg(\left\langle \nabla_{e}f(z^{(t-1)}),(s_{eu}\ s_{ev})\right\rangle +\left\Vert (s_{eu}\ s_{ev})-(z_{eu}^{(t-1)}\ z_{ev}^{(t-1)})\right\Vert _{2}^{2}\Bigg)$

return $z^{(T)}$
\end{algorithm}

We state without proof the following theorem which is similar to the
Algorithm \ref{alg:acdm}, whose proof also follows similarly from
\citet{ene2015random}.
\begin{thm}
Algorithm \ref{alg:rcdm} (option 1) and the fractional peeling procedure
\citep{harb2022faster} output an $\epsilon$-approximate dense decomposition
in $O\left(mn^{2}\log\frac{n}{\epsilon}\right)$ time in expectation.
\end{thm}

\section{Additional Experiment Results\label{sec:Additional-Experiment-Results}}

\subsection{Data summary}

We use eight datasets to be consistent with previous works, eg. \citet{boob2020flowless,harb2022faster}:
cit-Patents, com-Amazon, com-Enron, dblp-author, roadNet-CA, roadNet-PA,
wiki-topcats from SNAP collection \citet{leskovec2014snap} and orkut
from Konect collection \citet{kunegis2013konect}. We remark, however,
that road networks datasets (roadNet-CA, roadNet-PA) are expected
to be close to planar graphs, and therefore have very low maximum
density.

\begin{table}[H]
\centering{}{\small{}\caption{{\small{}Summary of datasets}}
\label{table:datasets}}%
\begin{tabular}{lcc}
\hline 
{\small{}Dataset} & {\small{}No. vertices} & {\small{}No. edges}\tabularnewline
\hline 
\hline 
{\small{}cit-Patents} & {\small{}3774768} & {\small{}16518947}\tabularnewline
\hline 
{\small{}com-Amazon} & {\small{}334863} & {\small{}925872}\tabularnewline
\hline 
{\small{}com-Enron} & {\small{}36692} & {\small{}367662}\tabularnewline
\hline 
{\small{}dblp-author} & {\small{}317080} & {\small{}1049866}\tabularnewline
\hline 
{\small{}roadNet-CA} & {\small{}1965206} & {\small{}5533214}\tabularnewline
\hline 
{\small{}roadNet-PA} & {\small{}1088092} & {\small{}3083796}\tabularnewline
\hline 
{\small{}wiki-topcats} & {\small{}1791489} & {\small{}25444207}\tabularnewline
\hline 
{\small{}orkut} & {\small{}3072441} & {\small{}117185083}\tabularnewline
\hline 
\end{tabular}
\end{table}

\subsection{Additional plots}

\begin{figure}[H]
{\small{}}\subfloat[wiki-topcats]{{\small{}\includegraphics[width=0.33\columnwidth]{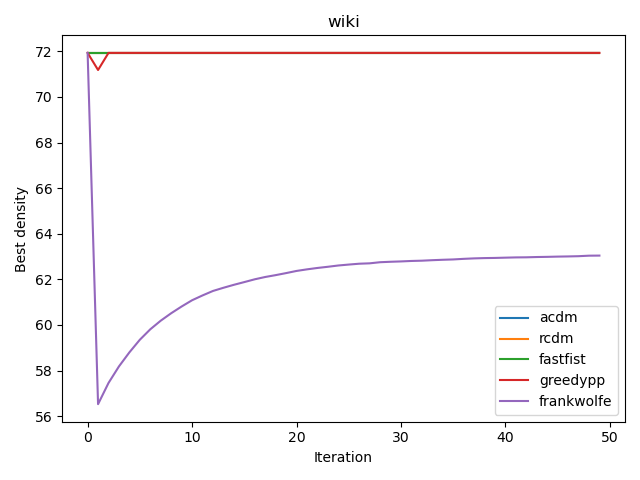}}{\small\par}

{\small{}}{\small\par}}{\small{}}\subfloat[roadNet-PA]{{\small{}\includegraphics[width=0.33\columnwidth]{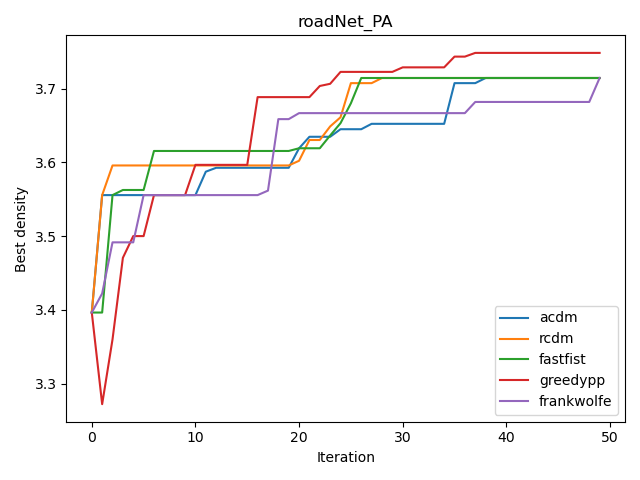}}}{\small{}}\subfloat[roadNet-CA]{{\small{}\includegraphics[width=0.33\columnwidth]{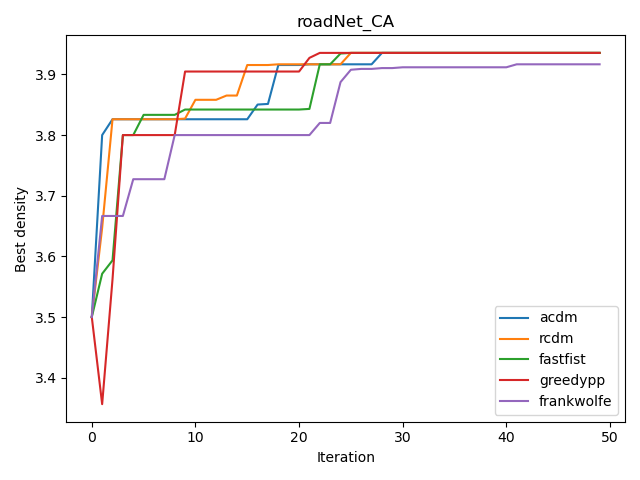}}}{\small{}\hfill{}}{\small\par}

{\small{}}\subfloat[cit-Patent]{{\small{}\includegraphics[width=0.33\columnwidth]{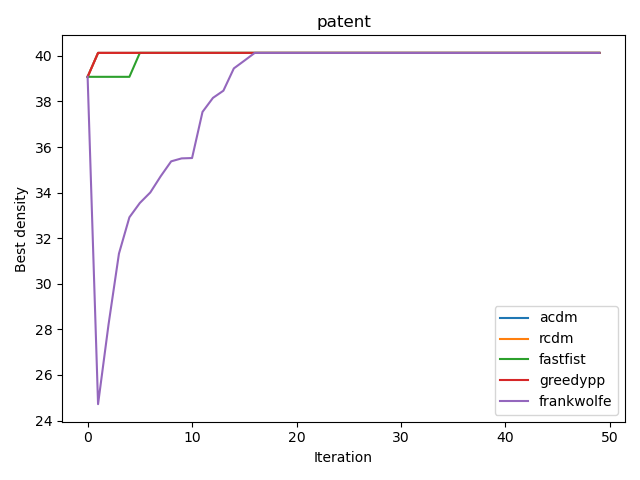}}{\small\par}

{\small{}}{\small\par}}{\small{}}\subfloat[com-Enron]{{\small{}\includegraphics[width=0.33\columnwidth]{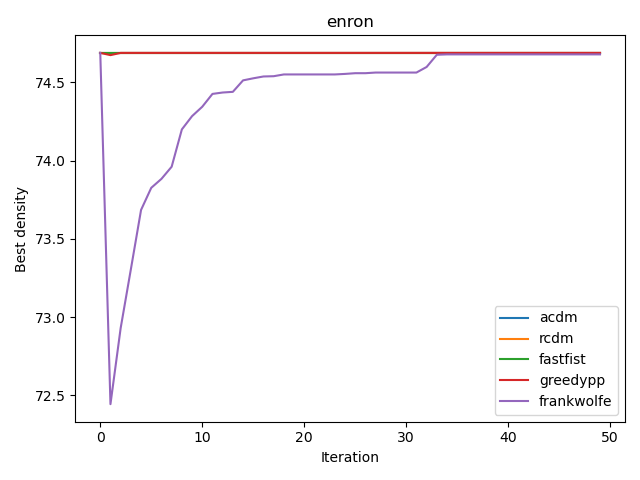}}}{\small{}}\subfloat[dblp-author]{{\small{}\includegraphics[width=0.33\columnwidth]{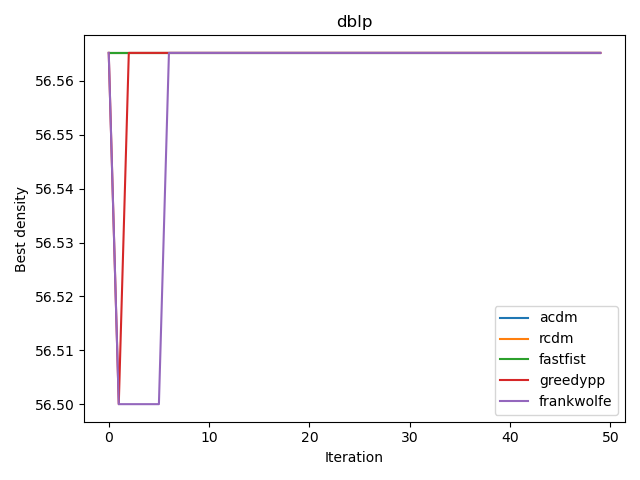}}}{\small{}\caption{{\small{}Iteration/Best Density}}
\label{fig:iteration-best-density-1}}{\small\par}
\end{figure}

\begin{figure}
{\small{}}\subfloat[wiki-topcats]{{\small{}\includegraphics[width=0.33\columnwidth]{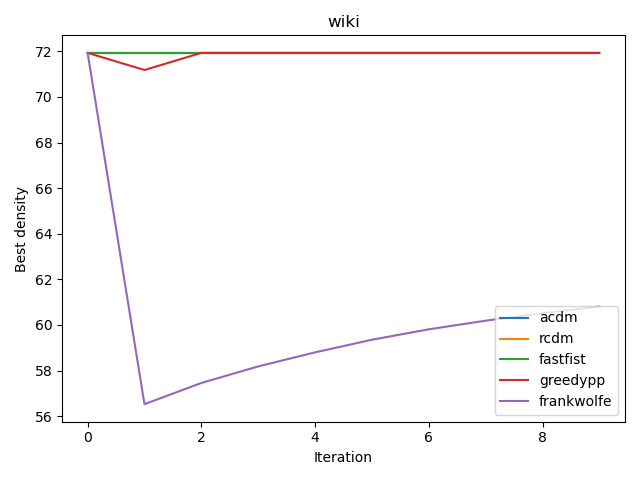}}{\small\par}

{\small{}}{\small\par}}{\small{}}\subfloat[roadNet-PA]{{\small{}\includegraphics[width=0.33\columnwidth]{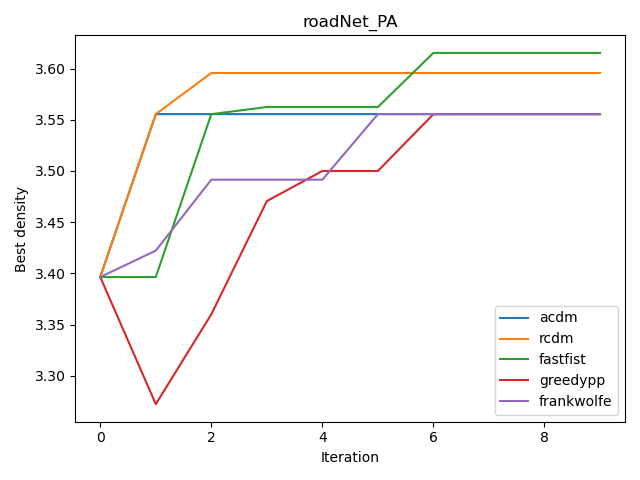}}}{\small{}}\subfloat[roadNet-CA]{{\small{}\includegraphics[width=0.33\columnwidth]{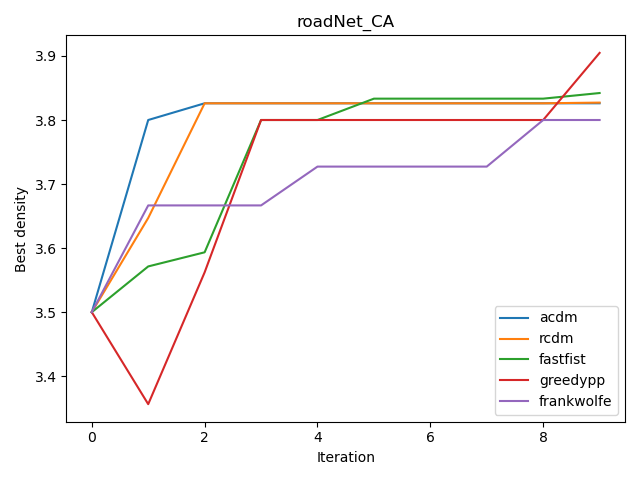}}}{\small{}\hfill{}}{\small\par}

{\small{}}\subfloat[cit-Patent]{{\small{}\includegraphics[width=0.33\columnwidth]{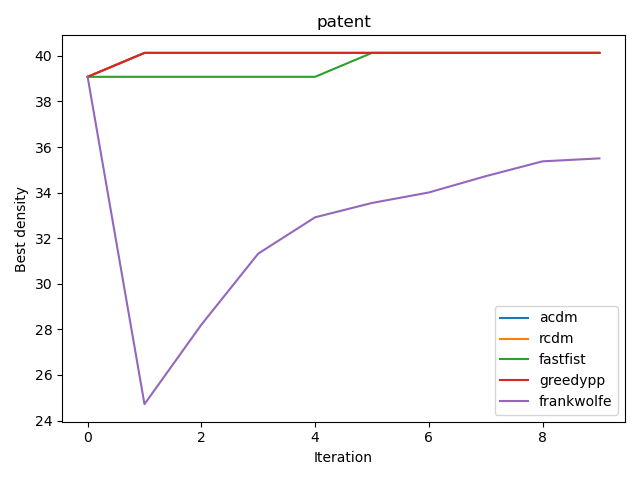}}{\small\par}

{\small{}}{\small\par}}{\small{}}\subfloat[com-Enron]{{\small{}\includegraphics[width=0.33\columnwidth]{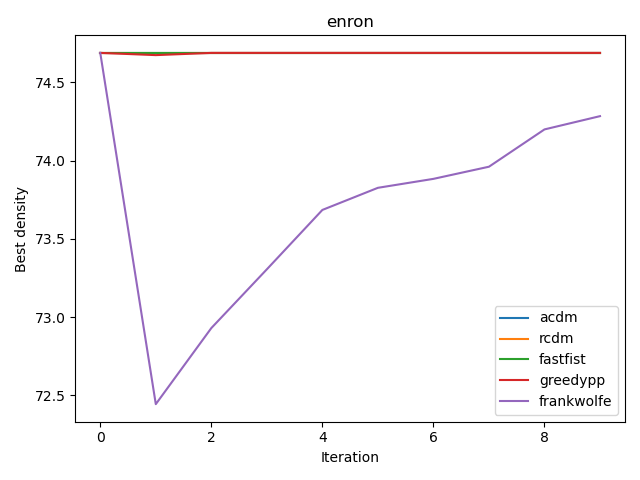}}}{\small{}}\subfloat[dblp-author]{{\small{}\includegraphics[width=0.33\columnwidth]{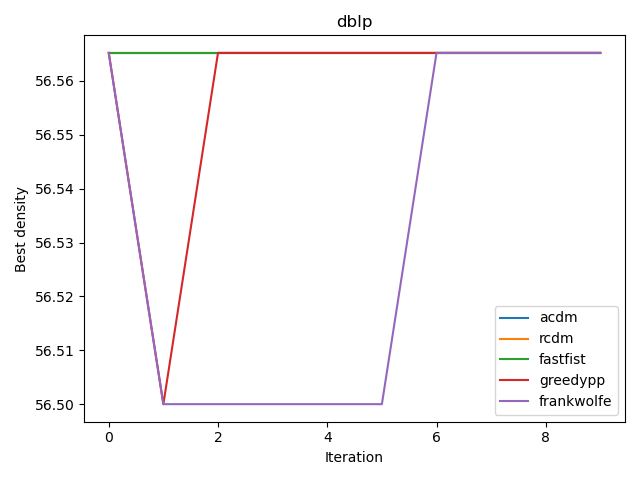}}}{\small{}\caption{{\small{}Iteration/Best Density zoomed in the first 10 iterations}}
\label{fig:iteration-best-density-1-zoom}}{\small\par}
\end{figure}

\begin{figure}
{\small{}}\subfloat[wiki-topcats]{{\small{}\includegraphics[width=0.33\columnwidth]{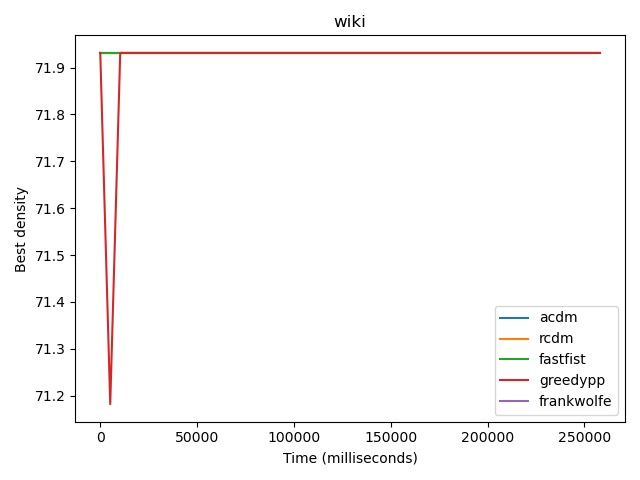}}{\small\par}

{\small{}}{\small\par}}{\small{}}\subfloat[roadNet-PA]{{\small{}\includegraphics[width=0.33\columnwidth]{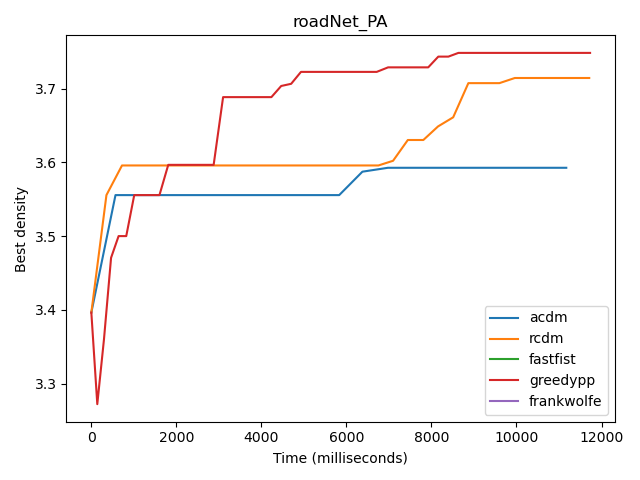}}}{\small{}}\subfloat[roadNet-CA]{{\small{}\includegraphics[width=0.33\columnwidth]{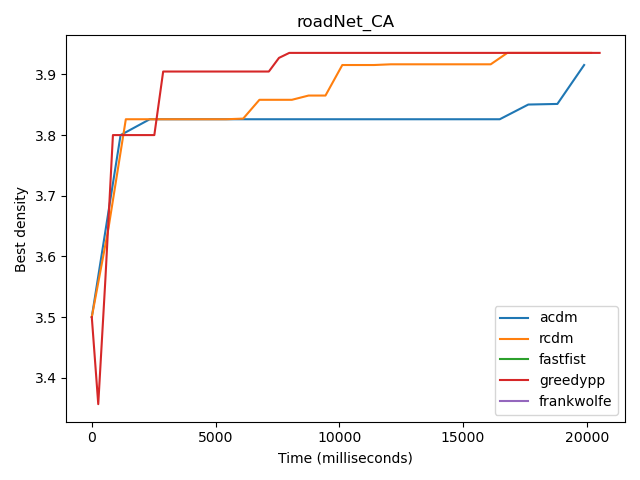}}}{\small{}\hfill{}}{\small\par}

{\small{}}\subfloat[cit-Patent]{{\small{}\includegraphics[width=0.33\columnwidth]{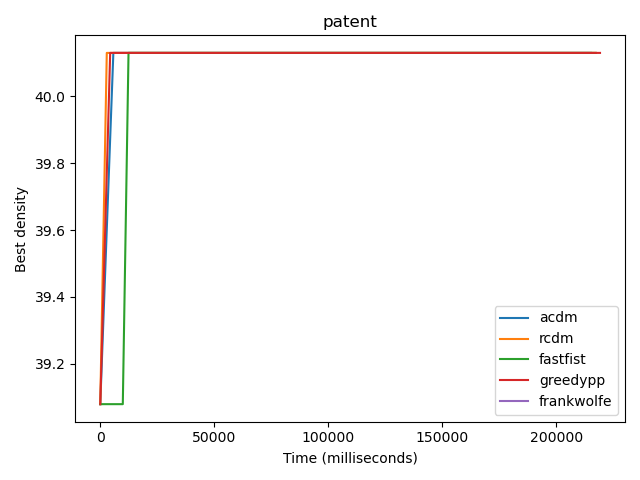}}{\small\par}

{\small{}}{\small\par}}{\small{}}\subfloat[com-Enron]{{\small{}\includegraphics[width=0.33\columnwidth]{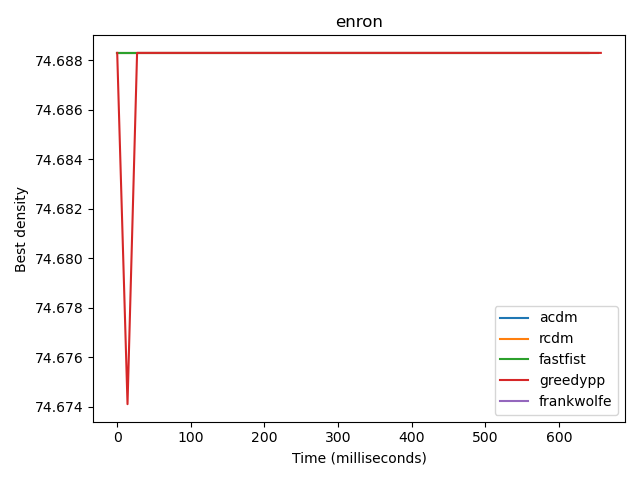}}}{\small{}}\subfloat[dblp-author]{{\small{}\includegraphics[width=0.33\columnwidth]{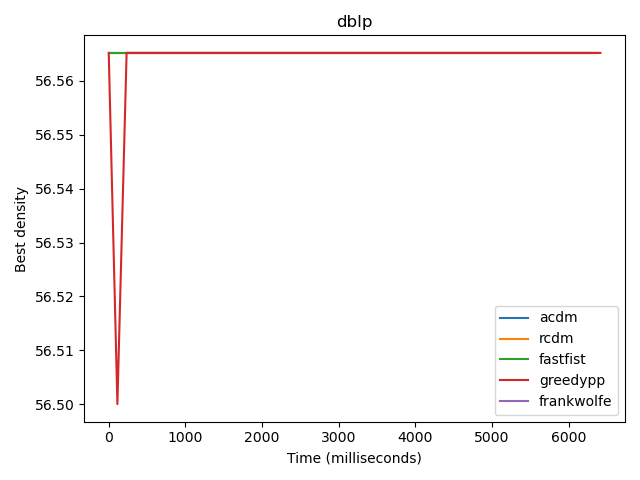}}}{\small{}\caption{{\small{}Time/Best Density}}
\label{fig:time-best-density-1}}{\small\par}
\end{figure}

\begin{figure}
{\small{}}\subfloat[wiki-topcats]{{\small{}\includegraphics[width=0.33\columnwidth]{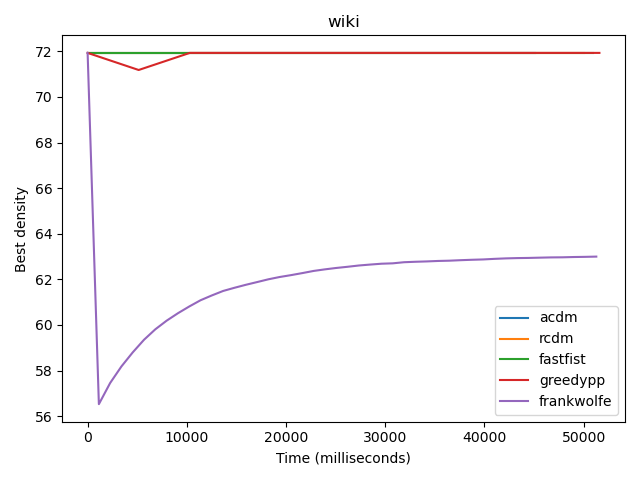}}{\small\par}

{\small{}}{\small\par}}{\small{}}\subfloat[roadNet-PA]{{\small{}\includegraphics[width=0.33\columnwidth]{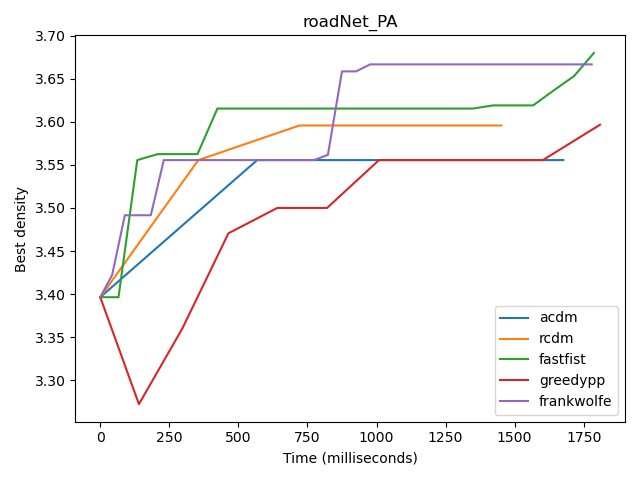}}}{\small{}}\subfloat[roadNet-CA]{{\small{}\includegraphics[width=0.33\columnwidth]{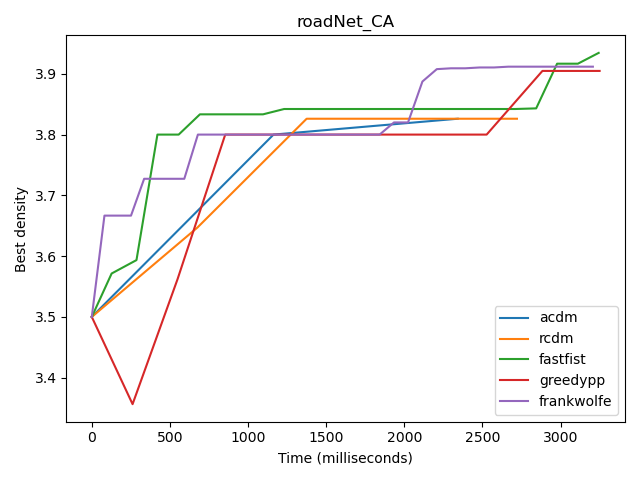}}}{\small{}\hfill{}}{\small\par}

{\small{}}\subfloat[cit-Patent]{{\small{}\includegraphics[width=0.33\columnwidth]{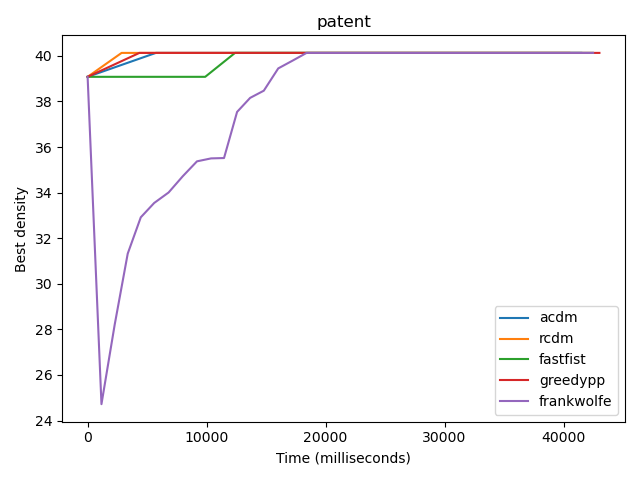}}{\small\par}

{\small{}}{\small\par}}{\small{}}\subfloat[com-Enron]{{\small{}\includegraphics[width=0.33\columnwidth]{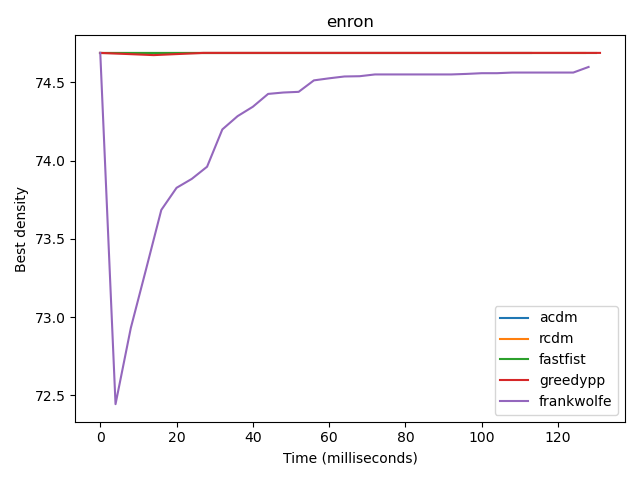}}}{\small{}}\subfloat[dblp-author]{{\small{}\includegraphics[width=0.33\columnwidth]{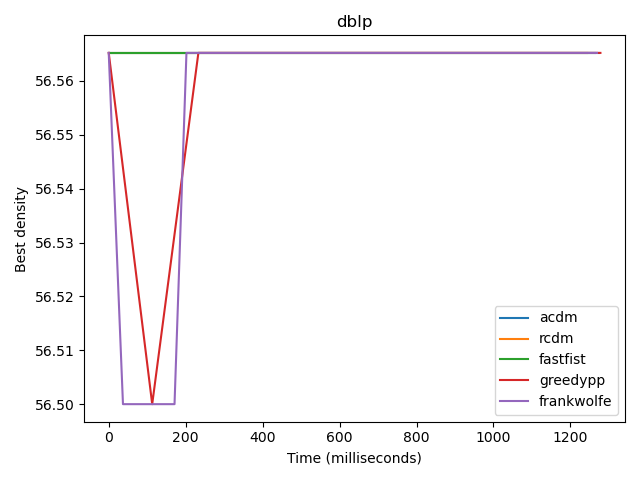}}}{\small{}\caption{{\small{}Time/Best Density zoomed in the first iterations}}
\label{fig:time-best-density-1-zoom}}{\small\par}
\end{figure}

\begin{figure}
{\small{}}\subfloat[wiki-topcats]{{\small{}\includegraphics[width=0.33\columnwidth]{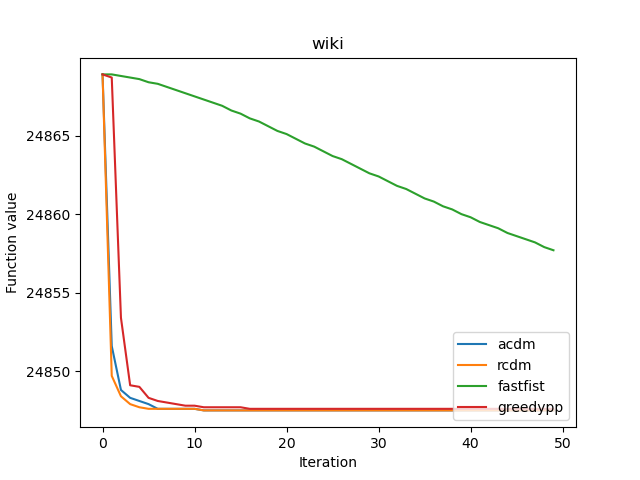}}{\small\par}

{\small{}}{\small\par}}{\small{}}\subfloat[roadNet-PA]{{\small{}\includegraphics[width=0.33\columnwidth]{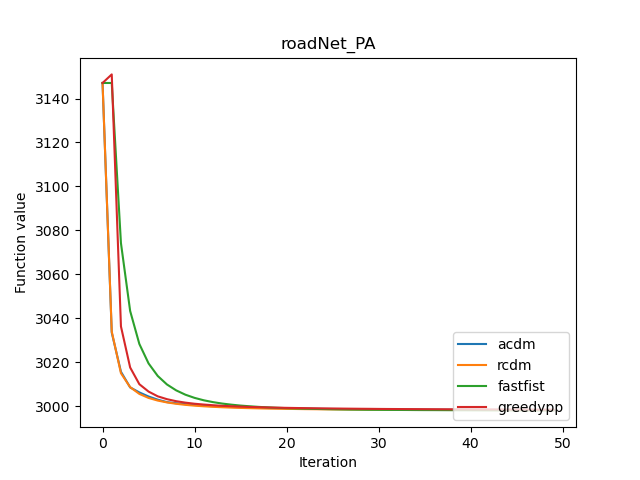}}}{\small{}}\subfloat[roadNet-CA]{{\small{}\includegraphics[width=0.33\columnwidth]{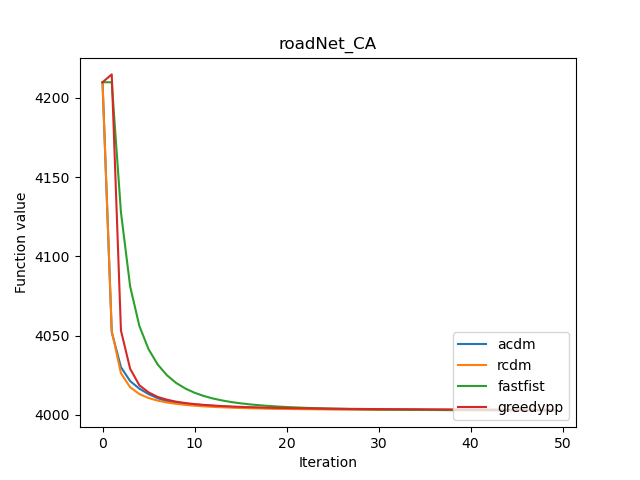}}}{\small{}\hfill{}}{\small\par}

{\small{}}\subfloat[cit-Patent]{{\small{}\includegraphics[width=0.33\columnwidth]{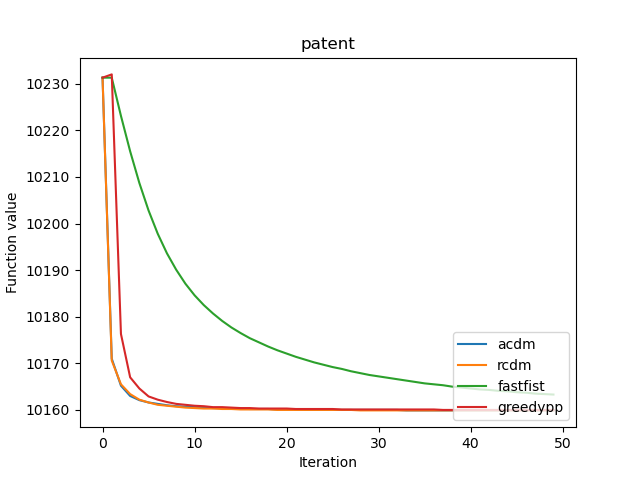}}{\small\par}

{\small{}}{\small\par}}{\small{}}\subfloat[com-Enron]{{\small{}\includegraphics[width=0.33\columnwidth]{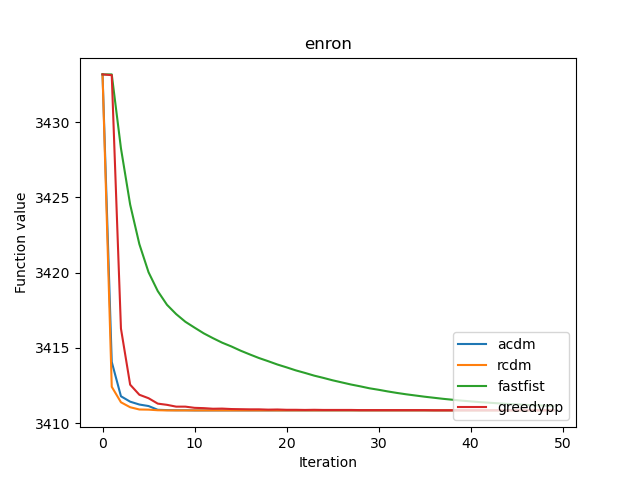}}}{\small{}}\subfloat[dblp-author]{{\small{}\includegraphics[width=0.33\columnwidth]{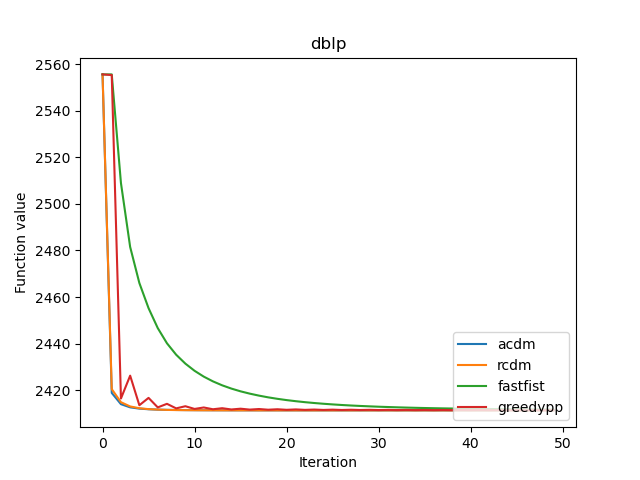}}}{\small{}\caption{{\small{}Iteration/Load norm $\left(\sum_{u\in V}b_{u}^{2}\right)^{1/2}$.
We exclude Frank-Wolfe from this plot as it performs significantly
worse than the other algorithms}}
\label{fig:iteration-function-value-1}}{\small\par}
\end{figure}

\end{document}